\documentclass[a4paper,11pt]{article}
\RequirePackage{natbib}

\usepackage{amsmath}
\usepackage{amssymb}
\usepackage{latexsym}
\usepackage{ntheorem}
\usepackage{ifthen}
\usepackage{xcolor}
\usepackage{newcent}
\usepackage{authblk}
\usepackage{dsfont}
\usepackage{hyperref}
\usepackage[title]{appendix}
\usepackage{cleveref}

\usepackage[margin=3cm]{geometry}

\theoremstyle{plain}

\theoremheaderfont{\bf}\theorembodyfont{\sl}

\newtheorem{theorem}{Theorem}[section]
\newtheorem{proposition}[theorem]{Proposition}
\newtheorem{lemma}[theorem]{Lemma}
\newtheorem{corollary}[theorem]{Corollary}

\theoremheaderfont{\it}\theorembodyfont{\rm}

\newtheorem{remark}[theorem]{Remark}

\theoremheaderfont{\bf}\theorembodyfont{\rm}

\newtheorem{definition}[theorem]{Definition}
\newtheorem{example}[theorem]{Example}

\newtheorem{assumption}[theorem]{Assumption}
\newtheorem{condition}[theorem]{Condition}

\theoremstyle{nonumberplain}

\theoremheaderfont{\bf}\theorembodyfont{\sl}

\newenvironment{proof}[1][]
{\ifthenelse{\equal{#1}{}}{\smallskip\noindent\textsl{Proof. }}{\smallskip
\noindent\textsl{Proof #1. }}}{\hfill$\Box$}

\newcommand{\RR}{{\rm I\kern-2pt R}}

\newcommand{\E}{{\rm I\kern-2pt E}}
\newcommand{\EE}{{\rm I\kern-2pt E}}

\newcommand{\PP}{{\rm I\kern-2pt P}}
\newcommand{\NN}{{\rm I\kern-2pt N}}

\def \N{\mathbb{N}}
\def \R{\mathbb{R}}

\newcommand{\QQ}{\mathbb{Q}}

\def\Xt{\pi_t(\xi)}
\def\XT{\pi_T(\xi)}
\def\Ac{{\cal A}}
\def\Bc{{\cal B}}
\def\Cc{{\cal C}}

\def\Ec{{\cal E}}
\def\Fc{{\cal F}}
\def\Fut{{\cal F}^{\cal U}_{t}}
\def\FuT{{\cal F}^{\cal U}_{T}}

\def\Pc{{\cal P}}
\def\Qc{{\cal Q}}
\def\Sp{\Ac}
\def\Sc{{\cal S}}

\def\Tr{{\cal H}(\mathbb{F}^{\mathcal{U}})}
\def\Uc{{\cal U}}

\def\Vt{V^{x,H,C}}

\def \o{\omega}
\def \O{\Omega}

\def\1{\mbox{\bf 1}}

\begin{document}

\title{The robust superreplication problem: a dynamic approach}

\author[1]{Laurence Carassus}
\author[2\footnote{ Support from the European Research Council under the European Union's Seventh Framework Programme (FP7/2007-2013) / ERC grant agreement no. 335421 is gratefully acknowledged. JO and JW are also thankful to St.\ John's College in Oxford for its financial support. JW further acknowledges support from the German Academic Scholarship Foundation.}]{Jan Ob\l{}\'oj}
\author[2$^*$]{Johannes Wiesel}
\affil[1]{L\'eonard de Vinci P\^ole Universitaire, Research Center
and LMR, Universit\'e de Reims-Champagne Ardenne.
Email: laurence.carassus@devinci.fr}
\affil[2]{Mathematical Institute and St. John’s College, University of Oxford, Oxford}

%\date{}

\maketitle

\begin{abstract}
  In the frictionless discrete time financial market of Bouchard et al.(2015) we consider a trader who, due to regulatory requirements or internal risk management reasons, is required to hedge a claim $\xi$ in a risk-conservative way relative to a family of probability measures $\Pc$.  We first describe the evolution of $\pi_t(\xi)$ - the superhedging price at time $t$ of the liability $\xi$ at maturity $T$ - via a dynamic programming principle and show that $\pi_t(\xi)$  can be seen as a concave envelope of $\pi_{t+1}(\xi)$ evaluated at today's prices. 
Then we consider an optimal investment problem for the trader who is rolling over her robust superhedge and phrase this as a robust maximisation problem, where the expected utility of inter-temporal consumption is optimised subject to a robust superhedging constraint. This utility maximisation is carrried out under a new family of measures $\Pc^u$, which no longer have to capture regulatory or institutional risk views but rather represent trader's subjective views on market dynamics. Under suitable assumptions on the trader's utility functions, we show that optimal investment and consumption strategies exist and further specify when, and in what sense, these may be unique.
\end{abstract}

\section{Introduction}

We consider a discrete time financial market and an agent who needs to hedge a liability $\xi$ maturing at a future date $T$ in a robust and risk-conservative way. Our focus is on the interplay between the beliefs used for assessing the risks, the beliefs used for agent's investment decisions and the dynamics of agent's actions. For simplicity we assume away other factors and consider an agent who can trade in a dynamic way with no constraints or frictions in $d$ assets available in the market at prices which are exogenous. More precisely, following the approach of \cite{Samuelson} and \cite{BS73}, risky assets are modelled as stochastic processes and their behaviour specified by a probability measure. However, unlike the classical uni-prior approach which fixes one such measure $P$, we consider a multi-prior framework and work simultaneously under a whole family of measures $P\in \Pc$. This offers a robust approach which accounts for model ambiguity, also referred to as \emph{Knightian uncertainty} after \cite{Kni}. 

The price to pay for a robust modelling view comes through specificity of outputs: while the uni-prior setting might generate a unique fair price for a derivate contract a multi-prior setting will typically generate a relatively wide interval of no-arbitrage prices, a tradeoff first identified in the seminal paper of \cite{merton_no_dominance1973}. 
We consider a trader who, due to regulatory requirements or internal risk management reasons, is required to hedge $\xi$ in a risk-conservative way relative to $\Pc$. This means that initially she has to allocate capital equal to $\pi(\xi)$, the superhedging price of $\xi$, i.e., the price of cheapest trading strategies which are guaranteed to cover the liability $\xi$ under all $P\in \Pc$. There might be many such cheapest superhedging strategies and the trader can pick any one of them to follow until time $T$. This is a conservative and non-linear risk assessment: the capital the trader would be allowed to borrow against a long position in $\xi$ is $-\pi(-\xi)$ and is typically significantly lower than $\pi(\xi)$.

The superhedging price $\pi(\xi)$ can be characterised theoretically and has been considered in a number of papers, see \cite{BN} and the discussion below. To the best of our knowledge, the focus of most of these works has been on the static problem: the problem today for the horizon $T$. In contrast, in this paper we want to focus on the dynamics of the robust pricing and hedging problem \emph{through time}. We ask how $\pi(\xi)$ changes \emph{over time} and how the trader should act optimally through time. Clearly, tomorrow she will see new prices in the market and will be able to recompute the superhedging price. If the new price is lower, she will be able to unwind her old position, buy a new position and be left with a surplus. She could then consume this (e.g., pay into her credit line if the initial capital was borrowed) or invest further if she believes the market offers suitable opportunities. 

Our first main contribution is to describe the evolution of $\pi_t(\xi)$ - the superhedging price at time $t$ of the liability $\xi$ at maturity $T$. We work in the setting of \cite{BN} and consider an abstract  set of priors $\Pc$, possibly large and in particular not dominated by a single probability measure. The measures $P\in \Pc$ are represented as compositions of one-step kernels and to establish the dual characterisation of $\pi_0(\xi)$ \cite{BN} have essentially proven a dynamic programming principle for the dual objects. We prove that $(\pi_t(\xi))_{0\leq t\leq T}$ satisfy a dynamic programming principle, and that $\pi_t(\xi)$ can be seen as a concave envelope of $\pi_{t+1}(\xi)$ evaluated at today's prices. To the best of our knowledge, this was first suggested in the robust setting by \cite{Dupire}. 
We also characterise $\pi_t(\xi)$ as the wealth of a minimal superhedging strategy in the sense of \cite{FK97}. These results provide natural robust extensions of classical uni-prior results, see \cite{fs}, including a robust version of the algorithm in \cite{CGT06}. Further, considering $\Pc$ which corresponds to the pointwise robust setting of \cite{Bur16b}, we show that $\pi_t(\xi)$ corresponds to the uniprior superhedging price for an extreme $P\in \Pc$. Proving our results in the robust setting requires rather lengthy and technical arguments. This is mainly due to delicate measurability questions.

Our second main contribution is to consider an optimal investment problem for a trader who is rolling over her robust superhedge. 
This is phrased as a problem of robust maximisation of expected utility of inter-temporal consumption subject to a robust superhedging constraint. Here the robust constraint means the superhedging has to be satisfies $P$-a.s.\ for all $P\in \Pc$. The robust utility maximisation means that we consider a max-min problem, where minimisation is over $P\in \Pc^u$. We argue that the latter problem should be considered with respect to a different set of priors $\Pc^u\subseteq \Pc$ than the former problem. Measures $P\in \Pc^u$ no longer have to capture regulatory or institutional risk views but rather represent trader's subjective views on market dynamics. Under suitable assumptions on the trader's utility functions, we show that optimal investment and consumption strategies exist and further specify when, and in what sense, these may be unique. We provide examples to illustrate various pitfalls occurring when our assumptions are not satisfied.

Throughout, we work in the setup of \cite{BN} who extended the classical uni-prior theory of pricing and hedging in discrete time to the robust mutli-prior case, introducing a suitable notion of no-arbitrage, proving a robust version of the fundamental theorem of pricing and hedging and establishing a robust pricing-hedging duality. Numerous authors have since adopted their setup and worked on robust extensions of the classical problems in quantitative finance such as pricing and hedging of American options, utility maximisation or transaction cost theory to name just a few examples, see \cite{Nutz,BC16,Aksamitetal:18,bayraktar2017arbitrage,bouchard2017super} and the references therein. We note that alternative ways to address model uncertainty are possible, including the pathwise, or pointwise, approach developed in \cite{DaHo07,ABPW13,ClassS,Burz16a,Bur16b} among others. Whilst the resulting robust framework for pricing and hedging is equipped with different notions of arbitrage and different fundamental theorems, it was recently shown by \cite{jw} to be equivalent to the multi-prior approach. Thus, on an abstract level, there is no loss of generality in our choice to adopt the multi-prior approach of \cite{BN}. It is important however that we work in discrete time. While in the classical setup no-arbitrage theory, including dynamic understanding of the superhedging price, is well developed in continuous time, see \cite{FK97,DelSch05}, in the robust setting an extension of abstract no-arbitrage theory, as developed in \cite{BN} or \cite{Bur16b}, to the continuous time is still open. This is despite a body of works which have achieved either particular or generic steps towards such a goal, large enough so that we can not do it justice in this introduction but refer to \cite{AP95,Ly95,DenisMartini:06,CO11, denis2007utility, EJ13,biagini2014robust,HO18,BCHPP:17,bartl2017pathwise} and the references therein.

We note that $d$ may be large and our assets may include both primary and derivate assets. Indeed, one way of making robust outputs more specific is by including more traded assets in the analysis. This was the original motivation behind the works on the robust pricing and hedging in continuous time, going back to \cite{Ho982}, where one typically assumes that the market prices of European options on the underlying assets co-maturing with our liability $\xi$ are known. Here, we consider an abstract general setup and allow any $d$-tuple of traded assets, for a finite $d$. We may expect that the level of uncertainty regarding different assets may differ and this would be reflected in $\Pc$. However it is crucial that all the assets are traded dynamically. From a theoretical standpoint, this is both necessary to obtain a dynamic programming principle for the superhedging prices and without loss of generality in the sense that any \cite{BN} setup where some assets are only available for trading at time $0$ can be lifted to a setup with dynamic trading in all assets in a way which does not introduce arbitrage and does not affect time-$0$ superhedging prices, see \cite{Aksamitetal:18}. From a practical standpoint, this is not a  significant assumption as we may only consider liquidly traded assets.

The remainder of the paper is organised as follows. The next section introduces and discusses our modelling framework. \Cref{sec. conc} presents the results characterising the dynamics of the superhedging price. We then specialise, in \cref{secex}, to the pathwise setting when $\Pc$ contains all measures with specified supports. This allows for a more intuitive interpretation of the results, easier proofs and explicit examples. \Cref{sec:utility} then considers the secondary utility maximisation problem for a trader who dynamically re-balances her superhedging strategy and states the existence and uniqueness results for the optimal investment and consumption strategies. Finally, proofs are presented in three appendices.

\section{Models of Financial markets}
\label{setup}

In this section we set up the multi-prior modelling framework and give introductory definitions. Future dynamics of financial assets are modelled using probability measures but, unlike the classical case where one such measure is fixed, we typically work simultaneously under all $P$ from a large family of measures $\Pc$. Our market has $d$ traded assets, these could be stocks or options, but importantly all are traded dynamically. We do not consider statically traded assets, i.e., only available for buy-and-hold trading, as then the superhedging prices typically can not admit a dynamic programming principle across all times, see \cite{Aksamitetal:18}.  
\subsection{Uncertainty modelling}
\label{ts}
We work in the setting of \cite{BN} to which we refer for details and motivation. 
We only recall the main objects of interest here and refer to \cite{bs}[Chapter 7] for technical details.
Let $\Omega$ be a Polish space and denote by $\Omega^t$ its $t$-fold Cartesian product. We define the price process $S$ of discounted prices of $d$ traded stocks as a Borel measurable map $S_t(\o)=(S_t^1(\o), \dots, S_t^d(\o)):\O^T \to \R^d_+$ for every $\o=(\o_0,\ldots,\o_T)$ with the convention $S_0(\omega)=s_0 \in \R^d_+$ and $T\in \mathbb{N}$ is the time horizon. Prices are specified in discounted units and we have a riskless asset with price equal to $1$ for all $0\le t \le T$. 
%The level sets of $S$ for $t\in \{0, \dots, T\}$ and $\omega \in \Omega^T$ are denoted by
%\begin{align*}
%\Le= \{ \tilde{\omega} \in \Omega^T \ | \ S_{0:t}(\omega)=S_{0:t}(\tilde{\omega}) \},
%\end{align*}
%where $S_{0:t} := (S_0, \dots S_t)$. 
Furthermore let $\mathfrak{P}(\Omega^t)$ be the set of all probability measures on $\mathcal{B} (\Omega^t)$, the Borel-$\sigma$-algebra on $\Omega^t$. We denote by $\Fut$ the universal completion of $\mathcal{B}(\Omega^t)$. We  often consider $(\Omega^t, \Fut)$ as a subspace of $(\Omega^T, \FuT)$ and write $\mathbb{F}^{\mathcal{U}}= (\Fut)_{t=0, \dots, T}$. In the rest of the paper, we will use the same notation for $P \in \mathfrak{P}(\Omega^T)$ and for its (unique) extension to $\FuT$.
For a given $\Pc \subseteq  \mathfrak{P}(\Omega^T)$, a set $N \subset \Omega^T$ is called a $\Pc$-polar  if for all $P \in \Pc$, there exists some $A_{P} \in \mathcal{B}(\Omega^T)$ such that $P(A_{P})=0$ and $N \subset A_{P}$. We say that a property holds $\Pc$-quasi-surely (q.s.), if it holds outside a $\Pc$-polar set. Finally we say that a set is of $\Pc$-full measure  if its complement is a $\Pc$-polar set.\\
To give a probabilistic description of the market we consider  a family of random sets  $\Pc_{t} : \Omega^t \twoheadrightarrow \mathfrak{P}(\O)$, for all $0\leq t\leq T-1$.  The set $\Pc_{t}(\omega)$ can be seen as the set of all possible models for the $t+1$-th period given the path $\omega \in \Omega^t$ at time $t$. In order to aggregate trading strategies on different paths in a measurable way, we assume here that the sets $\Pc_t$ have the following property:
\begin{assumption}
\label{Qanalytic}
The set $\Pc$ has Analytic Product Structure (APS), which means that
\begin{align*}
\Pc=\{P_0 \otimes \cdots \otimes P_{T-1} \ | \ P_t \text{ is an}\  \Fut\text{-measurable selector of }\Pc_t \},
\end{align*}
where the sets $\Pc_t(\omega) \subseteq \mathcal{P}(\Omega)$ are nonempty, convex and 
\begin{align*}
\text{graph}(\Pc_t) = \{ (\omega, P)\ | \ \omega \in \Omega^t, \ P \in \Pc_t (\omega)\}
\end{align*}
is analytic.
\end{assumption}

%In the continuous-time case, a dominated set of priors $\Pc$ can arise when there is uncertainty on the drift of the underlying process  while non-dominated set of priors may arise if there is uncertainty  on the volatility of this process (see e.g. \cite{EJ13}). In the case of volatility uncertainty,  the corresponding set is however weakly compact (see for instance  \cite{AP95},  \cite{Ly95}, \cite[Proposition 3]{DenKer}  and also \cite{EJ13}). We propose concrete examples of non-compact random sets $\Pc_{t}(\omega)$ in Section \cref{secex}.
%Furthermore we note that instead of specifying a set of probability measures $\Pc$, modelling of different agents' beliefs can be achieved by selecting a set $X \subseteq \Omega$ of possible price trajectories. This is often referred to as the pathwise approach to modelling financial markets. In \cite{jw} it is shown that the pathwise and the quasisure approach are essentially equivalent as long certain technical assumptions on $X$ and $\Pc$ are satisfied. Here we decide to work in a quasisure setting, but will discuss pathwise examples in Section \cref{secex}, which naturally arise in practical applications. For a background on pathwise arbitrage and superhedging we refer to \cite{Bur16b,Burz16a,ClassS}.

The fact that $\text{graph}(\mathcal{P}_t)$ is analytic allows for an application of the Jankov-von-Neumann theorem (\cite[Prop. 7.49, p.182]{bs}), which guarantees the existence of universally measurable selectors $P_t: \O^t \to \mathfrak{P}(\O)$. Here $P_0 \otimes \cdots \otimes P_{T-1}$ denotes the $T$-fold application of Fubini's theorem, which defines a measure on $\mathfrak{P}(\Omega^T)$.
Indeed, analyticity of the graph of $\Pc_t$ is of paramount importance for the preservation of measurability properties. For example the proof of a quasisure superreplication theorem (see \cite[Lemma 4.10]{BN}) uses the fact that if  $X_{t+1}: \Omega^{t+1} \to \mathbb{R}$  is upper semianalytic, then $\sup_{P \in \mathcal{P}_{t}(\omega)} \E_{P} [X_{t+1}(\o,\cdot)]$ remains upper semianalytic.
Apart from \cref{Qanalytic}, we make no specific assumptions on the set of priors $\Pc$. It is neither assumed to be dominated by a given reference probability measure nor to be weakly compact. 
Some concrete examples, including when $\Pc_{t}(\omega)$ are non-compact random sets, are discussed in \cref{secex}.
\subsection{Trading}
Trading strategies are represented by $\mathbb{F}^{\mathcal{U}}$-predictable $d$-dimensional  processes $H:=\{H_{t}\}_{ 1 \le t \le T}$ where for all $1 \leq t \leq T$, $H_{t}$ represents the investor's holdings in  each of the $d$ assets at time $t$. The set of trading strategies is denoted by $\Tr$.  Investors are allowed to consume and their cumulative consumption is represented by an $\R$-valued $\mathbb{F}^{\mathcal{U}}$-adapted process $C=\{C_{t} \}_{1 \le t \le T}$, $C_{0}=0$ and which is assumed to be non-decreasing: $C_{t} \leq C_{t+1}$ $\Pc\mbox{-q.s.}$ The set of cumulative consumption processes is denoted by $\Cc$. We will use the notation
$\Delta S_{t}=S_{t}-S_{t-1}$ and $\Delta C_{t}=C_{t}-C_{t-1}$ for $1\le t\le T.$
Given an initial wealth $x\in \R$, a trading portfolio $H$ and a cumulative consumption process $C$, the wealth process $\Vt$ is governed by
\begin{align}\label{hj}
\Vt_0 &=x  \nonumber \\ 
\Vt_t &=\Vt_{t-1}+H _{t} \Delta S_{t}-\Delta C_{t}\quad\mbox{for}\;1\le t\le T.  
\end{align}
The condition $C=0$ means that the portfolio $H$ is self-financing and in this case we  write $V^{x,H}$ instead of $V^{x,H ,0}$.

We are interested in superhedging of a (European) contingent claim and therefore adapt the presentation of \cite{FK97} to the robust framework. A (European) contingent claim is represented by an $\FuT$-measurable random variable $\xi$ and the set of superhedging strategies for $\xi$ is denoted by
\begin{align}
\label{aah} 
\Sp(\xi):=\left\{ (x,H ,C) \in \R \times \Tr \times \Cc \ \bigg| \ \Vt_T \geq \xi \ \Pc\mbox{-q.s.} \right\}.
\end{align}
\begin{definition}
 The superreplication price $\pi(\xi)$ of an $\FuT$-measurable random variable $\xi$ is the minimal initial capital needed for superhedging $\xi$, i.e., 
 \begin{align}\label{piG}
\pi(\xi):=\inf \left\{x \in \mathbb{R} \ | \ \exists (H ,C)\in \Tr \times \Cc \mbox{ such that } (x,H,C) \in \Sp(\xi) \right\},
\end{align}
with $\pi(h)=+\infty$ if $\Sp(\xi) = \emptyset$. A superhedging strategy $(\hat{x},\hat{H},\hat{C}) \in \Sp(\xi)$ is called \emph{minimal} if for all $(x,H ,C) \in \Sp(\xi)$ $\Vt_t \ge V_{t}^{\hat{x},\hat{H},\hat{C}}$ $\Pc$-q.s. for all $0\le t\le T$. 
\end{definition}
It is easy to see that $\hat{x}=\pi(\xi)$ for any minimal superhedging strategy $(\hat{x},\hat{H},\hat{C}) \in \Sp(\xi)$.
%We now recall \cite[Superhedging Theorem, Theorem 2.3]{BN} for the convenience of the reader.
%\begin{theorem}
%\label{BN2}
%Let NA$(\Pc)$ hold and let $\xi$ be an $\FuT$-measurable random variable. Then $\pi(\xi)>-\infty$ and there exists a trading strategy $H \in \Tr$ such that $V_{T}^{\pi(\xi),H} \geq \xi$ $\Pc$-q.s.
%Moreover if Assumption \cref{Qanalytic} holds true and $\xi$ is upper semianalytic, we have
%\begin{align}
%\label{repdual}
%\pi(\xi)=\sup_{Q \in \Qc} \EE_{Q} (\xi).
%\end{align}
%\end{theorem}
%This easily implies the following lemma:

%\begin{lemma}
%Let $\xi$ be an $\FuT$-measurable random variable. Let $(\hat{x},\hat{H},\hat{C})\in \Sp(\xi)$ be a minimal superhedging strategy.  Then $\hat{x}=\pi(\xi)$.
%\end{lemma}
%\begin{proof}
%For any $n>0$, by definition of $\pi(\xi)$, we find $(\pi(\xi)+1/n,H,0) \in \Sp(\xi)$ and thus, by definition of minimal superhedging strategies, $\hat{x} \leq \pi(\xi)+1/n$. Letting $n\to\infty$, we have $\hat{x}\leq \pi(\xi)$. Conversely
%$(\hat{x},\hat{H},\hat{C}) \in \Sp(\xi)$ and by definition of $\pi(\xi)$ we conclude that $\pi(\xi)\leq \hat{x}$.
%\end{proof}

\subsection{No-arbitrage condition and Pricing measures}
\label{BNexp}
We recall the no-arbitrage condition introduced in \cite{BN}.
\begin{assumption}\label{NAQT}
There is no $\Pc$-quasisure arbitrage (NA$(\Pc)$) in the market if for all $H  \in \Tr$ with $V_{T}^{0,H} \geq 0 \ \Pc\mbox{-q.s.}$ we have $V_{T}^{0,H}  = 0 \ \Pc\mbox{-q.s. }$
\end{assumption}
The above definition gives an intuitive extension of the classical no-arbitrage condition, specified under a fixed probability measure $P$, to the multi-prior case of family of probability measures $\Pc$. The intuition is justified by the FTAP generalisation proved by \cite[Theorem 4.5]{BN}: under \cref{Qanalytic} (recall that $S$ is Borel-adapted) NA$(\Pc)$ is equivalent to the fact that for all $P \in \Pc$, there exists some $Q \in \Qc$ such that $ P \ll Q$ where
\begin{align}\label{mathR}
 \Qc:=\{Q \in \mathfrak{P}(\Omega^T)\ | \ \exists \, P \in \Pc,\ Q \ll P \; \mbox{and $S$ is a martingale under $Q$}\}.
\end{align}
\begin{remark}
 By the same token, further results, e.g., on the Superhedging Theorem or the worst-case expected utility maximisation (see \cite{Nutz},  \cite{BC16},  \cite{Bart16} and \cite{NS16}) provide more evidence supporting the view that NA$(\Pc)$ is a well-chosen extension of the classical no-arbitrage assumption. However, the price to pay when using NA$(\Pc)$ is related to technical measurability issues arising when one considers a one step version of the NA$(\Pc)$ (see \eqref{NP1} below). In \cite{Bart16} a stronger version of \cref{NAQT} is introduced which states that \eqref{NP1} below is satisfied for all $\o \in \O^t$. In \cite{BC16}, a stronger version of no-arbitrage is proposed (sNA$(\Pc)$) which states that there is no-arbitrage in the classical sense for all measures $P \in \Pc$. In both cases  some of the measurability issues are simplified. Finally, different approaches to model uncertainty may lead to fundamentally different notions of arbitrage. In the pathwise approach, one typically asks that some subset of paths supports a feasible model -- this is in contrast to the multi-prior setup in this paper where essentially \emph{all} $P\in \Pc$ are assumed to be feasible models. In consequence, the no-arbitrage conditions in the pathwise approach, e.g., model independent arbitrage as in \cite{DaHo07,CO11,ABPW13} or  Arbitrage de la classe $\Sc$  (see \cite{ClassS}), are much weaker than NA$(\Pc)$, i.e., their notions of arbitrage are much stronger than the $\Pc$-q.s.\ arbitrage. To wit, negation of sNA$(\Pc)$ above gives that there is a classical arbitrage for at least one $P\in \Pc$ while \cite{DaHo07} say that there is a \emph{weak arbitrage opportunity} if there is a classical arbitrage under \emph{all} $P\in \Pc$.
\end{remark}

The one step version of the NA$(\Pc)$ is the following: for $\omega \in \Omega^t$ fixed we say that NA$(\Pc_{t}(\omega))$ condition holds if for all $H \in \mathbb{R}^d$
\begin{align} \label{NP1}
H\Delta S_{t+1}(\omega,\cdot) \geq 0 \; \Pc_{t}(\omega) \mbox{-q.s.} \quad\Rightarrow \quad H\Delta S_{t+1}(\omega,\cdot) = 0 \;\Pc_{t}(\omega)\mbox{-q.s.} \end{align}
It is proved in \cite[Theorem 4.5]{BN} that under the assumption that $S$ is Borel measurable and (APS) of $\Pc$, the condition NA$(\Pc)$ is equivalent to the fact that for all $0\leq t\leq T-1$, there exists some $\Pc$-full measure set $\Omega^{t}_{NA} \in \Fut$, such that for all $\omega \in \Omega^{t}_{NA}$, NA$(\Pc_{t}(\omega))$ holds. We also introduce the one-step versions of the set $\Qc$:
\begin{align*}
\Qc_{t}(\omega)=\left\{Q \in \mathfrak{P}(\Omega)\ | \ \exists \, P\in \Pc_{t}(\omega) \mbox{ such that } Q \ll P\; \mbox{ and } \EE_Q [\Delta S_{t+1}(\omega,\cdot)] =0\right\}.
\end{align*}

As is shown in \cite[Lemma 4.8]{BN}, $\Qc_{t}$ has an analytic graph. An application of the Jankov-von Neumann Theorem and Fubini's Theorem shows that we have
\begin{align}\label{Mstar}
\Qc= \{ Q_0 \otimes \cdots \otimes Q_{T-1} \ | \ Q_t \mbox{ is }\Fut\mbox{-measurable selector of }\Qc_t \text{ for all }0\le t\le T-1 \}.
\end{align}
%\begin{small}
%\begin{align}
%\label{Mstar}
%\widehat{\mathcal{M}}^{t}:=\{ Q_{1}\otimes q_{2} \otimes \dots \otimes q_{t},\;  Q_{1} \in \mathcal{M}_{1},\;  q_{s+1} \in \mathcal{S}K_{s+1},
%  q_{s+1}(\cdot,\o^{s}) \in \mathcal{M}_{s+1}(\o^{s})\; Q_{s} \mbox{-a.s.} \; \,  s \in \{1,\dots, t-1\}  \},
%\end{align}
%\end{small}
%Then looking carefully to the proof of \cite[Theorem 4.5]{BN}, we get that for all $t \in \{,\ldots,T\}$
%\begin{align}
%\widehat{\mathcal{M}}^{t}= \mathcal{M}^{t}=\{P \in \mathfrak{P}(\O^{t}),\; \exists \, Q^{'} \in \mathcal{Q}^{t}, P \ll Q^{'} \; \mbox{and $P$ is a martingale measure}\}.
% \end{align}

\section{Existence and characterisation of minimal superhedging strategies}
\label{sec. conc}
The Superhedging theorem, also known as the pricing-hedging duality, is one of the fundamental results in the classical setting of $\Pc=\{P\}$, see \cite{fs,FK97} and the references therein. One of the main results in \cite{BN} was its extension to the multi-prior case: 
\begin{align}
\label{repdual}
\pi(\xi)=\sup_{Q \in \Qc} \EE_{Q} [\xi].
\end{align}
While this duality is important and theoretically pleasing, its use for computations may be hampered by lack of a tractable characterisation of the set $\Qc$. One of our aims is to give a more algorithmic approach to the above duality. To this end, we establish a suitable dynamic programming principle (DPP) for the superhedging price and also show existence of minimal superhedging strategies in the spirit of \cite{FK97}. This leads to a robust generalisation of the algorithm in \cite{CGT06} and gives a way to handle computation of superhedging prices and, importantly, strategies.

\subsection{Main Result}
To state our main result we need to introduce some further notation. For an upper semianalytic function $\xi: \Omega^T \to \R$ let $\{\Xt\}_{0\le t\le T}$ denote the one step superhedging prices $\pi_t(\xi): \Omega^t \to \overline{\R}$ given by
\begin{equation}
\begin{split}
\XT(\omega)  &=  \xi(\omega),\quad \textrm{ and for }0 \leq t \leq T-1
\label{defminr}\\ 
\Xt(\omega) &= \inf \{x \ | \ \exists H \in \R^d \text{ such that } x + H \Delta S_{t+1}(\omega,\cdot) \ge \pi_{t+1}(\xi)(\omega,\cdot) \; \Pc_{t}(\omega)\mbox{-q.s.}\}. 
\end{split}
\end{equation}
Note that the above superhedging prices can be construed as concave envelopes. Indeed, with a slight abuse of notation we denote the one-step quasisure concave envelope $\widehat{f}: \Omega^{t}\times \R_+^d \to \R$ by
\begin{align*}
\widehat{f}(\omega,s)=\inf \{u(s) \ | \ u:\R^d_+ \to \R \mbox{ closed concave, }u(S_{t+1}(\o,\cdot)) \ge f(\omega,\cdot) \ \Pc_t(\omega)\text{-q.s.} \}
\end{align*}
for $t\in \{1, \dots, T\}$ and an upper semianalytic function $f:\Omega^{t}\times \Omega\to\R$, where we recall that a concave function is closed, if its superlevel set is closed. As every concave function can be written as the pointwise infimum of linear functions the equality
\begin{equation}\label{eq:sh-cncenv}
\Xt(\omega)=\widehat{\pi_{t+1}(\xi)}(\omega,S_{t}(\omega)),\quad \omega\in \Omega^t,\quad 0\leq t\leq T-1
\end{equation}
holds and the one-step superhedging prices can be obtained by iteratively taking concave envelopes in the coordinates of $\Omega$.

Let us now define the corresponding dual expressions for the one step case. For $\omega \in \Omega^t$ and $f: \, \O^t \times \O \to \overline{\R}$, we define $\Ec_t (f) : \, \Omega^t  \to \overline{\R}$ by
\begin{align*}
\Ec_t(f) (\o)= \sup_{Q \in \Qc_{t}(\omega)}\EE_Q[f(\omega, \cdot)].
\end{align*}
Furthermore, for measurable $\xi: \Omega^T \to \R$, we define the sequences of operators
\begin{align}
\Ec^T(\xi)  =  \xi\quad \textrm{ and }\quad 
\label{defect} \Ec^t(\xi) =  \Ec_t \circ  \Ec^{t+1} (\xi),\quad 0 \leq t \leq T-1.
\end{align}
With notation at hand, we can state our first main result which gives existence of minimal superhedging strategies and establishes a Dynamic Programming Principle for $\Xt$ and $\Ec^t(\xi)$.
\begin{theorem}
\label{eu} Let \cref{Qanalytic} and NA$(\Pc)$ hold.
Let $\xi:\Omega^T \to\R$ be  an upper semianalytic function such that $\sup_{Q\in\mathcal{Q}} \EE_Q[\xi^- ] < \infty.$ Then:
\begin{itemize}
\item[(i)] there exists a minimal superhedging strategy in $\Sp(\xi)$;
\item[(ii)] for any minimal superhedging strategy
$(\hat{x},\hat{H},\hat{C})\in \Sp(\xi)$, its value satisfies
\begin{align}
\label{eeu} V^{\hat{x},\hat{H},\hat{C}}_t &= \Xt=\Ec^t(\xi) \;\; \Pc \mbox{-q.s.},\quad 0\leq t\leq T.
\end{align}
In particular,
\begin{align*}
\hat{x}=\pi(\xi) &= \pi_0(\xi)=\Ec^0 (\xi).
\end{align*}
\end{itemize}
\end{theorem}
Perhaps suprisingly the proof of the above result is technically involved and is thus relegated to Appendix \ref{appendix:laurence}. However in the special case of the canonical setting $\Omega=\R_+^d$, $S_t(\omega)=\omega_t$ and $\Pc=\{P \in \mathfrak{P}(X) \ | \ \text{supp}(P) \text{ is finite}\}$ for an analytic set $X \subseteq \Omega^T$  the underlying arguments are quite intuitive and simple. We outline them in the next section.

\subsection{Canonical space: Concave envelopes and computation of the superhedging price}
\label{secex}

In this subsection we work on the canonical space, i.e. we set $\Omega= \R^d_+$ and $S_t(\o)=(\o_t^1, \dots, \o_t^d)$. In particular $\xi(S_1(\o), \dots, S_T(\omega))=\xi(\o)$ holds.\\
We start by developing in more detail the special case when $\Pc$ is obtained by specifying the support for feasible moves of the stock prices. This captures the pathwise approach but is also natural in the quasisure framework as NA$(\Pc)$ and $\pi(\xi)$ only depend on the polar sets of $\Pc$. More precisely we give the following definition:
\begin{definition}
Assume that for $0\le t\le T-1$ we are given correspondences $f_t: \Omega^t \twoheadrightarrow \R^d$. We say that a sequence of sets $(\Pc_t)_{0\le t \le T-1}$ such that $\Pc_t \subseteq \mathfrak{P}(\Omega)$ for all $0 \le t \le T-1$ is generated by $\{f_t\}_{0 \le t \le T-1}$ if 
\begin{align*}
\Pc_t(\omega)=\{ P \in \mathfrak{P}(\Omega) \ | \ \text{supp}(P) \subseteq f_t(\omega) \}
\end{align*}
for $0\le t \le T-1$, where $\text{supp}(P)$ denotes the support of a measure $P$. 
\end{definition}
Recall that a correspondence $f: \Omega^t \twoheadrightarrow \R^d$ is called measurable if $\{\omega \in \Omega^t \ | \ f(\omega) \cap O \neq \emptyset\}\in \mathcal{B}(\Omega^t)$ for all open sets $O \subseteq \R^d$. We refer to \cite[14.A, p.643ff.]{rw} for the theory of measurable correspondences.
%For such a set $\Pc$ the following measurability property holds:
\begin{lemma}\label{lem. graph}
Let $(\Pc_t)_{0\le t \le T-1}$ be generated by measurable, closed valued correspondences $\{f_t\}_{0 \le t\le T-1}$. Then $\Pc_t$ has Borel measurable graph for all $0 \le t \le T-1$.
\end{lemma}
Under the assumptions of \cref{lem. graph} we can then define $\Pc \subseteq \mathfrak{P}(\Omega^T)$ satisfying (APS) as in \cref{Qanalytic}.\\
\begin{proof}
By assumption the graph of $f_t$ is $\Bc(\Omega^t) \otimes \Bc(\R^d)=\Bc((\R^d)^{t+1})$-measurable for all $t \in \{0, \dots T-1\}$ (see \cite[Theorem 14.8, p.648]{rw}). Thus by \cite[Cor. 7.25.1, p.134]{bs} $\mathfrak{P}(\text{graph}(f_t))$ is Borel as well. Define the map
\begin{align*}
D: \Omega^t \times \mathfrak{P}(\R^d_+) \to \mathfrak{P}(\Omega^{t+1}), \ (\omega, P) \mapsto \delta_{\omega} \otimes P
\end{align*}
and note that $D$ is a homeomorphism from $\Omega^t \times \mathfrak{P}(\R^d_+)$ to $\{\delta_{\omega} \otimes P \ |\ \omega \in \Omega^t, \ P \in \mathfrak{P}(\R^d_+) \}$. Indeed, take a sequence $(\omega_n, P_n) \in \Omega^t \times \mathfrak{P}(\R^d_+)$ such that $(\omega_n, P_n)$ converges to $(\omega, P)$ in the product topology. Denote by $\mathcal{L}^1_b(\Omega^{t+1})$ the bounded 1-Lipschitz functions on $\Omega^{t+1}$. Then
\begin{align*}
&\lim_{n \to \infty}\sup_{f \in \mathcal{L}^1_b(\Omega^{t+1})} \left|\int_{\Omega^{t+1}} f d(\delta_{\omega_n} \otimes P_n) - \int_{\Omega^{t+1}} f d(\delta_{\omega} \otimes P) \right| \\
\le\ &\lim_{n \to \infty} \left(|\omega_n-\omega|+\sup_{f \in \mathcal{L}^1_b(\Omega^{t+1})} \left|\int_{\Omega^{t+1}} f(\omega, \cdot) dP_n - \int_{\Omega^{t+1}} f(\omega, \cdot)  dP \right|\right)=0,
\end{align*}
so $\delta_{\omega_n} \otimes P_n$ converges weakly to $\delta_{\omega}\otimes P$. Continuity of the inverse map follows directly from the definition of weak convergence of measures. Note also that a homeomorphism map Borel sets to Borel sets. As
\begin{align*}
\mathfrak{P}(\text{graph}(f_t)) \cap \{\delta_{\omega} \otimes P \ |\ \omega \in \Omega^t,\ P \in \mathfrak{P}(\R^d) \}
\end{align*}
is Borel-measurable, applying the inverse map $D^{-1}$ we conclude that
\begin{align*}
\text{graph}(\Pc_t)=D^{-1}(\mathfrak{P}(\text{graph}(f_t)) \cap \{\delta_{\omega} \otimes P \ |\ \omega \in \Omega^t,\ P \in \mathfrak{P}(\R^d) \})
\end{align*}
is Borel.
\end{proof}

In fact, for such a set $\Pc$ the condition NA$(\Pc_t(\omega))$ is equivalent to $0 \in \text{ri}(f_t(\omega)-S_t(\omega))$, where $\text{ri}(A)$ denotes the relative interior of the convex hull of $A$. For a proof of this result in a more general setup, see \cite[Thm. 3.3, p. 6]{jw}. This deterministic condition is called No Pointwise Arbitrage in \cite{Bur16b} and can be checked without resorting to the use of probability measures. 

As an intuitive outline of the proof of \cref{eu}, let us now assume that $\Pc=\{P \in \mathfrak{P}(X) \ | \ \text{supp}(P) \text{ is finite}\}$ and NA$(\Pc)$ holds, where $X \subseteq \Omega^T$ is some analytic set. 
%\marginpar{\textcolor{red}{Do we need NA(\Pc) or at least $X^*\neq \emptyset$?}}
We can now prove the crucial equality $\Xt=\Ec_t(\pi_{t+1}(\xi))$ directly using the concave envelope characterisation \eqref{eq:sh-cncenv}, see also \cite{beiglbock2014martingale} and the references therein. Indeed, it follows from \cite[Prop 6.1, p. 14]{jw} that $\Pc$ satisfies \cref{Qanalytic} in this case and
\begin{align*}
\Qc= \{ Q \in \mathfrak{P}(X) \ | \ \text{supp}(Q) \text{ is finite and }S \text{ is a martingale under }Q \},
\end{align*}
see also \cite[Example 1.2, p.827]{BN} for $X=(\R^d)^T$ and \cite[Cor. 4.6, p.151]{Lange} for locally compact $X$. Let $\omega=(\omega_1, \dots, \omega_t) \in \Omega^t$. Using Jensen's inequality
\begin{align}\label{eq. jensen}
\Ec_t(f)(\omega) &= \sup_{Q \in \Qc_t(\omega)} \EE_Q[f(\omega, \cdot)]\nonumber
\le \sup_{Q \in \Qc_t(\omega)} \EE_Q[\hat{f}(\omega, \cdot)]\\
&\le \sup_{Q \in \Qc_t(\omega)}\hat{f}(\omega, \EE_Q[\cdot])=\hat{f}(\omega, \omega_t),
\end{align}
where $\E_{Q}[\cdot]=\int_{\R_+^d} yQ(dy)$. To establish the $``\ge"$-inequality, it suffices to observe that  
\begin{align*}
s \mapsto \sup_{Q \ll P \text{ for some }P\in \Pc_t(\omega), \ \EE_{Q}[\cdot]=s} \EE_Q[f(\omega,\cdot)]
\end{align*}
is concave and dominates $f(\omega,\cdot)$ on $S_{t+1}(\Sigma_t^{\omega})$, where  $\Sigma_t^{\omega}:= \{\tilde{\omega}\in X \ | \ (\tilde{\omega}_1, \dots, \tilde{\omega}_t)=\omega \}$. While concavity is clear in general (see \cite[Lemma 2.2]{beiglbock2014martingale}), the domination property crucially relies on the fact that the set $\{Q \ll P \ \text{for some }P \in \Pc_t(\omega), \ \EE_{Q}[\cdot]=s\}$ contains the Dirac measures at points $s \in S_{t+1}(\Sigma_t^{\omega}).$ For a general set $\Pc$ this is not true: For example in the case $\Pc=\{P\}$ for some $P \in \mathfrak{P}(\Omega)$ in general only the set $ \{Q \ll P, \ \EE_{Q}[\cdot]=s\}$ is non-empty for $s$ in the relative interior of the convex hull of the support of $P$ (see \cite[Theorem 1.48, p.29]{fs}).

The following definition further characterises closed-valued correspondences $\{f_t\}_{0 \le t \le T-1}$ and is needed to identify an important subclass of sets $\{\Pc_t\}_{0 \le t \le T-1}$ generated by $\{f_t\}_{0 \le t \le T-1}$:

\begin{definition} \label{def. corunifcont}
A closed-valued correspondence $f_t: \Omega^t \to \R^d$ is called uniformly continuous if for all $\epsilon>0$ there exists $\delta >0$ such that for all $\omega, \omega' \in \Omega^T$ such that $|\omega'-\omega|\le \delta$ we have $d_H(f_t(\omega), f_t(\omega'))\le \epsilon$,
where $$d_H(A,B):=\max\left(\sup_{v \in A} \inf_{\tilde{v} \in B}|v- \tilde{v}|, \sup_{\tilde{v}\in B} \inf_{v \in A}|v-\tilde{v}| \right)$$ denotes the Hausdorff metric on closed subsets $A,B$ of $\Omega$.
%$f_t(\omega') \subseteq f_t(\omega)+B_{\epsilon}(0)$.
\end{definition}
Uniformly continuous correspondences are in particular continuous (see \cite[Def. 5.4, p.152]{rw}) and thus measurable (\cite[Theorem 5.7, p.154]{rw}). It turns out, that when the correspondences fulfil this continuity condition and are compact-valued, the $\Pc$-q.s. superhedging price of a continuous payoff $\xi$ coincides with the $P$-a.s. superhedging price of $\xi$ for every $P$ with support equal to the paths generated by the correspondences $\{f_t\}_{0 \le t \le T-1}$:
\begin{proposition}\label{prop:pathwiseex}
Suppose $(\Pc_t)_{0 \le t \le T-1}$ is generated by closed-valued, uniformly continuous correspondences $\{f_t\}_{0 \le t \le T-1}$ and that NA$(\Pc)$ holds. Furthermore assume that the function $\xi: \Omega^T \to \R$ is continuous and $\{f_t\}_{0\le t\le T-1}$ are compact-valued.
%\item[(ii)] the function $\xi: \Omega^T \to \R$ is continuous and bounded.
%\end{enumerate}
Take any measure $P=P_0 \otimes \cdots \otimes P_{T-1}$ such that 
\begin{align*}
\text{supp}(P_t(\omega))=f_{t}(\omega), \quad 0 \le t \le T-1, \ \omega \in \Omega^t.
\end{align*}
Then, for all $0 \le t \le T-1$ and $\omega \in \Omega^t$,
\begin{align}\label{eq. continuous}
\Xt(\omega)=\inf \{ x \in \R \ | \ \exists H \in \R^d \text{ such that } x+H \Delta S_{t+1}(\omega, \cdot) \ge \pi_{t+1}(\xi)(\omega, \cdot) \ P\text{-a.s}\}.
\end{align}
and $\omega \mapsto \pi_t(\xi)(\omega)$ is continuous.
\end{proposition}

%{\color{blue} Turns out it is not true for $\xi$ uniformly continuous (even Lipschitz) and non-compact $\{f_t\}_{0\le t \le T-1}$, even under additional regularity assumptions on $\{f_t\}_{0\le t \le T-1}$!}

The proof of the above result is relegated to Appendix \ref{appendix:A}. \\
We now apply this result to a one-dimensional case of particular interest, as in \cite{CV17}, where it is easy to explicitly compute the minimal superhedging prices:
\begin{proposition}
\label{pg1} Assume that for all $0\leq t\leq T-1$,  $d_{t+1}< 1 <  u_{t+1}$ and that  the (random) sets $\Pc_{t}$ are given by
\begin{align*}
\Pc_t(\omega)=\left\{P \in \mathfrak{P}(\R)\ |\ \mbox{supp}(P) \subset [\o_t d_{t+1},\o_t  u_{t+1}]  \right\},
\end{align*}
where $\omega=(\omega_1, \dots, \omega_t)\in \Omega^t$. Then NA$(\Pc)$ holds. Let $\xi: \R^T \to \R$ be convex.
%Let $h$ be  a convex  upper semianalytic, $\Fc_{T}$-measurable function from
%$\R^{T}$ into $\R$. 
%Assume that $
%\sup_{Q\in\mathcal{M}^{T}} \EE_Q(h^- ) < \infty.$ 
 Then
\begin{align}
\label{fgp}
 \XT  & =  \xi \nonumber\\
\Xt (\omega) & =  \alpha_{t+1} \pi_{t+1}(\xi)(\omega, \omega_{t} u_{t+1}) +  (1-\alpha_{t+1}) \pi_{t+1}(\xi)(\omega, \o_{t}
d_{t+1}),
\end{align}
where $\alpha_{t} := \frac{1-d_{t}}{u_{t}-d_{t}}$, $1\le t \le T$.
\end{proposition}
\begin{proof}
Noting that $f_t(\omega)=[\omega_t d_{t+1}, \omega_t u_{t+1}]$ is a uniformly continuous compact-valued correspondence, the graph of $\Pc_t$ is clearly non-empty, convex and Borel measurable for $0\le t \le T-1$ by \cref{lem. graph}. As $0 \in \text{ri}(f_t(\omega)-S_t(\omega))=\text{ri}([-\omega_t(1-d_{t+1}),\omega_t(u_{t+1}-1))$, NA$(\Pc)$ holds. 
%By Proposition \cref{prop:pathwiseex} we have
%\begin{align*}
%\pi_{T-1}(\xi)(\omega)&= \inf \{ x \in \R \ | \ \exists H \in \R \text{ such that } x+H\Delta S_{T}(\omega, \cdot) \ge \xi(\omega, \cdot) \text{ on } [\omega_{T-1} d_{T}, \omega_{T-1} u_{T}]\}.
%\end{align*} 
%for any continuous function $\xi: \R^{T} \to \R$.
%%, where
%%\begin{align*}
%%\widehat{\xi}(\omega,s)= \inf \{ u(s) \ | \ u: \R \to \R \text{ closed concave}, \ u(\cdot) \ge v(\omega, \cdot) \text{ on } [\omega_{T-1} d_{T}, \omega_{T-1} d_{T}]\}.
%%\end{align*}
%Now we can employ results from Theorem \cref{eu}.
We prove by induction that $\pi_t(\xi)$ satisfies \eqref{fgp} and is convex: This is clear for $t=T$. Now we assume that for some $0\le t \le T-1$, $\pi_{t+1}(\xi)$ is convex. 
%By continuity of $\pi_{t+1}(\xi)$ we conclude by \eqref{eq. continuous} for the uniform probability measure on $[\omega_{t} d_{t+1}, \omega_{t} u_{t+1}]$ that
As $\Pc_t(\omega)$ contains the Dirac measures on $[\omega_td_{t+1}, \omega_t u_{t+1}]$ we conclude that
\begin{align*}
\pi_{t}(\xi)(\omega)&= \inf \{ x \in \R \ | \ \exists H \text{ s. t. } x+H\Delta S_{t+1}(\omega, \cdot) \ge \pi_{t+1}(\xi)(\omega, \cdot) \text{ on } [\omega_{t} d_{t+1}, \omega_{t} u_{t+1}]\}.
\end{align*}
As $\pi_{t}(\xi)(\omega)$ is the pointwise concave envelope of the convex function $\pi_{t+1}(\xi)(\omega,\cdot)$, it can be written as the unique convex combination of the extreme points of $\pi_{t+1}(\xi)(\omega,\cdot)$ on the interval $[\omega_{t} d_{t+1}, \omega_{t} u_{t+1}]$, which conserves the barycentre $\omega_t$. Thus, we obtain \eqref{fgp} for $t$. Clearly $\Xt:\R^{t} \to \R$ is then a linear combination of convex functions (with non-negative coefficients) and thus also a convex function.
%So we get that
%\begin{eqnarray*}
%\Gamma_{t-1} h (\o^{t-1}) & = &
%\inf_{\{(\alpha,\beta) \in \R^2 \;| \; \alpha + \beta \o_{t-1} x \geq \Gamma_{t}
%h(\o^{t-1}, \o_{t-1} x), \; x \in \{d_{t}, u_{t}\} \}}
%\left\{\alpha + \beta \o_{t-1} \right\}\\
%& = & \bar \a + \bar \b \o_{t-1},
%\end{eqnarray*}
%where $(\bar \a,\bar \b)$ are the unique $(\a,\b)$ satisfying
%$ \alpha + \beta \o_{t-1} x = \Gamma_{t} h(\o^{t-1}, \o_{t-1} x)$ for both
%$x = d_{t}$ and $x=u_{t},$ i.e.
%$$ \left\{ \begin{array}{l}
% \bar \alpha = \frac{u_{t}\Gamma_{t} h(\o^{t-1}, \o_{t-1} d_t)-d_t  \Gamma_{t} h(\o^{t-1}, \o_{t-1} u_t)}{u_{t}-d_{t}} \\
% \bar \beta = \frac{\Gamma_{t} h(\o^{t-1}, \o_{t-1} u_t)-\Gamma_{t} h(\o^{t-1}, \o_{t-1} d_t)}{\o_{t-1}(u_{t}-d_{t})} \\
% \end{array} \right. $$
\end{proof}

It is insightful to observe that the above superreplication price corresponds to the actual replication price in a Cox-Ross-Rubinstein model of \cite{CRR79} where the stock price evolves on a binomial tree with $S_{t+1}\in \{d_{t+1}S_t,u_{t+1}S_t\}$. 
%
%The risky asset $S^{crr}$ is defined by:
%\begin{eqnarray*}
%S^{crr}_0 & = & S_0 , \\
%S^{crr}_{t+1} & = & S^{crr}_t U^{crr}_{t+1},\; \; \mbox{ for } t \in
% \{0,\ldots T-1\},
%\end{eqnarray*}
%where the independent random variables $U^{crr}_t$ are defined on the probability space
%$(\Omega^{crr}, 2^{\Omega^{crr}})$ by
%\begin{eqnarray*}
%U^{crr}_{t} ~:~ \Omega^{crr} := \prod_{t=1}^T\{u_t,d_t\} & \longrightarrow & \{u_t,d_t \}.\\
%%\o=(\o_t)_{t \in \{1,\ldots T\}} & \longrightarrow & \o_t
%\end{eqnarray*}
%We set  $\mathbb{F}^{crr}=\{\Fc^{crr}_t\;|\; t \in \{0, \dots, T\} \}$ with
%$\Fc_0=\{\emptyset, \Omega^{crr}\}$,
%$\Fc^{crr}_t=\sigma(U^{crr}_1,\ldots, U^{crr}_t)$ for $t \in
%\{1,\ldots T\}$ and we define the probability $P^{crr}$ on
%$2^{\Omega^{crr}}$ by
%\begin{align*}
%P^{crr}(\omega)=\prod_{t=1}^T \left( \alpha_t \delta_{\{u_t\}}(\o_t)  +(1-\alpha_t)\delta_{\{d_t\}}(\o_t) \right)
%\end{align*}
%
%Proposition \cref{pg1} can then be rewritten as follows.
%\begin{proposition}
%Under the assumptions of Proposition \cref{pg1},  we have for $t=0,\ldots, T$ that
%\begin{align*}
%\Xt(\omega) & =  \EE^{crr} \left( {\xi}(\o,\o_tU^{crr}_{t+1}, \ldots, \o_t U^{crr}_{t+1} \ldots U^{crr}_T) \right), \\
%\Xt(S^{crr}_1, \ldots, S^{crr}_t) & =  \EE^{crr} \left( {\xi}(S^{crr}_1, \ldots, S^{crr}_T) \; | \; \Fc^{crr}_t \right).
%\end{align*}
%\end{proposition}
%The super-replication price of $\xi$ is thus the replication price of ${\xi}(S^{crr}_1, \ldots, S^{crr}_T)$ in the Cox-Ross-Rubinstein model defined above.
%
\section{Maximising expected utility of consumption in $\Sp(\xi)$}\label{sec:utility}

\subsection{Main results}
In \cref{eu} above, we characterised the superhedging prices $\pi_t(\xi)$ and introduced ways for computing minimal superhedging strategies. However, these are typically non-unique. Indeed, as we see from \eqref{eq:sh-cncenv}, if the concave envelope $\widehat{f(\omega,\cdot)}$ of a function $f: \O^{t+1}\to \R$ is not differentiable at $\omega_t$, every point $H\in \R^d$ in its superdifferential constitutes a minimal superhedging strategy, see also \cref{Ex 1} below. 
To select the ``best" among minimal superhedging strategies we propose a secondary optimisation problem of robust maximisation of expected utility with intermediate consumption, given by
\begin{align} \label{eq. optpro}
\sup_{(H,C)\in \mathcal{A}_x} \inf_{P \in \Pc^u} \E_{P}\left[\sum_{s=1}^T U(s,\Delta C_s)\right],
\end{align}
where $\mathcal{A}_x$ is the set of investment-consumption strategies which superhedge $\xi:\Omega^T \to \R$, i.e.
\begin{align*}
\mathcal{A}_x :=\{ (H,C) \in \mathcal{H}(\mathbb{F}^{\mathcal{U}}) \times \mathcal{C}\ | \  V_T^{x,H,C}\ge \xi \ \Pc \text{-q.s.}\}
\end{align*}
and the set $\Pc^{u} \subseteq \mathfrak{P}(\Omega^T)$ fulfils the following condition:
\begin{assumption}\label{Ass 1b}
$\Pc^{u}$ satisfies (APS) and $\Pc^{u} \subseteq \mathcal{P}$.
\end{assumption}
%Note that if $U(T,\cdot)$ is increasing, then the condition $V_T^{x,H,C}= \xi$ in the definition of $\mathcal{A}_x$ can be assumed without loss of generality. 
The set $\Pc^{u}$ represents the subjective views of an investor. While superhedging with respect to $\Pc$ reflects the necessity to satisfy certain regulatory and risk requirements, $\Pc^u$ is used to express individual preferences for the optimisation problem \eqref{eq. optpro} and does not need to satisfy any further requirements than those of \cref{Ass 1b}, e.g. NA$(\Pc^{u})$ can fail. 
In \cref{thm. exis_new} and \cref{thm. unique_new} below, we show that \eqref{eq. optpro} is well posed and admits an optimiser which, under suitable assumptions, is unique.

The assumptions imposed on the utility functions $U(t, \cdot, \cdot)$ are in line with those in \cite{Nutz}:
\begin{assumption} \label{Ass 1a}
For $t=1, \dots, T$ the utility function $U(t,\cdot, \cdot): \Omega^t \times [0,\infty) \to \R$ is lower semianalytic and bounded from above. Furthermore 
\begin{enumerate}
\item $\omega \mapsto U(t,\omega, x)$ is bounded from below for each $x>0$.
\item $x \mapsto U(t,\omega,x)$ is non-decreasing, concave and continuous for each $\omega \in \Omega^t$.
\end{enumerate}
\end{assumption}
We believe that boundedness assumptions on utility functions which we make here could be weakened, similarly to \cite{BC16}. However, due to the overall length and already technical character of proofs, we decided to leave this extension for further research. \\
We remark that by 2. in \cref{Ass 1a} it is sufficient to consider investment-consumption strategies which hedge $\xi$, i.e. for which $V_T^{x,H,C}=\xi$, since the superhedging surplus can be consumed at terminal time.\\
Note that by \cref{Ass 1a} and standard results on Carath\'{e}odory functions (see \cite[Lemma 4.51, p. 153]{Hitch}) we conclude that $U(t, \cdot, \cdot)$ is $\mathcal{F}^{\mathcal{U}}_t\otimes \mathcal{B}(\R_+)$-measurable. We set $U(t,x,\omega)= -\infty$ for $x<0$ and often write $U(t,x)$ instead of $U(t,x,\omega)$.

\begin{theorem}\label{thm. exis_new} Let $U(t,\cdot,\cdot)$ be given for $1\le t \le T$ and let NA$(\Pc)$, \cref{Qanalytic}, \cref{Ass 1b} and \cref{Ass 1a} hold. Then for any Borel $\xi: \O^T\to \R$ such that $\sup_{Q\in\mathcal{Q}} \EE_Q[\xi^- ] < \infty$ there exists $(\hat{H}, \hat{C}) \in \mathcal{A}_{\pi}$ such that 
\begin{align*}
\inf_{P \in \Pc^u} \E_{P}\left[\sum_{s=1}^T U(s,\Delta\hat{C}_s)\right]=\sup_{(H,C)\in \mathcal{A}_{\pi}} \inf_{P \in \Pc^u} \E_{P}\left[\sum_{s=1}^T U(s,\Delta C_s)\right],
\end{align*}
where $\pi=\pi(\xi)$ is the $\Pc$-q.s. superhedging price of $\xi$.
\end{theorem}

In order to obtain uniqueness of the above maximiser $(\hat{H}, \hat{C})$, we again switch to the canonical setup $\Omega^T= (\R_+^d)^T$, $S_t(\omega)=\omega_t$. In line with \cite{denis2007utility} we strengthen assumptions on the utility functions $U(t, \cdot, \cdot)$ and also assume weak compactness of the set $\Pc^u$. This enables us to show existence of a ``worst-case" measure $\hat{P}\in\Pc^u$, in analogy to the argumentation in \cite{schied2005duality}. In fact, \cref{ex. non-unique} below shows, that one cannot expect uniqueness of maximizers in general, if $\Pc^u$ is not weakly closed.
\begin{assumption}\label{Ass 3} For $t=1, \dots, T$  the non-random utility functions $U(t, \cdot)$ satisfy \cref{Ass 1a} and are bounded. The mapping $x \mapsto U(t,x)$ is strictly concave, non-decreasing and continuous. Furthermore, for $t=0, \dots T-1$ and $\Pc^u$-q.e $\omega \in \Omega^t$ the set $\Pc^u_t(\omega)$ is weakly compact and  the sets $\Pc$ and $\Pc^u$ fulfil the following continuity criteria:
\begin{enumerate}
\item If $\omega, \tilde{\omega} \in \Omega^t$ and $\epsilon>0$, then there exists $\delta >0$ such that for $|\omega-\tilde{\omega}|\le \delta$ and for every $P \in \Pc^u_t(\omega)$ there exists $\tilde{P} \in \Pc^u_t(\tilde{\omega})$ such that $d_L(P, \tilde{P}) \le \epsilon$, where $$d_L(P, \tilde{P})=\inf\{\epsilon \ge 0 \ | \ P(A) \le \tilde{P}(A^\epsilon) + \epsilon \text{ for all }A \in \mathcal{B}(\Omega)\}$$ denotes the Levy metric on $\mathfrak{P}(\Omega)$ and $A^\epsilon=\{\omega \in \Omega \ | \ \exists \tilde{\omega}\in A \text{ such that }|\omega-\tilde{\omega}|< \epsilon \}$.
\item 
The map $f_t(\omega):=\text{supp}(\Pc_t(\omega))$ is uniformly continuous in the sense of \cref{def. corunifcont}, where
\begin{align*}
\text{supp}(\Pc_t(\omega))= \bigcap \{A \subseteq \Omega \text{ closed} \ \big| \  P(A)=1 \text{ for all } P \in \Pc_t(\omega)\}
\end{align*}
is the quasisure support of $\Pc_t(\omega)$ for $\omega \in \Omega^t$.
\end{enumerate}
\end{assumption}

\begin{theorem}\label{thm. unique_new}
In the setup of \cref{thm. exis_new} assume further that \cref{Ass 3} holds and that the functions $\pi_t(\xi): \Omega^t\to \R$ are continuous on $\{(\omega, v)\in  \Omega^t \ | \ v\in f_{t-1}(\omega)\}$ for all $1\le t\le T$. Then there exists a probability measure $\hat{P} \in \Pc^u$ such that 
\begin{align*}
\sup_{(H,C)\in \mathcal{A}_{\pi}} \inf_{P \in \Pc^u} \E_{P}\left[\sum_{s=1}^T U(s,\Delta C_s)\right]=\sup_{(H,C)\in \mathcal{A}_{\pi}} \E_{\hat{P}}\left[\sum_{s=1}^T U(s,\Delta C_s)\right].
\end{align*}
Furthermore, the maximising strategy $(\hat{H}, \hat{C}) \in \mathcal{A}_{\pi}$ is unique in the following sense: for any two maximising strategies $(H^1,  C^1), (H^2, C^2) \in \mathcal{A}_{\pi}$ and for $1\le t\le T$ we have $ C^1_t= C^2_t$ and $H_t^1 \Delta S_t= H_t^2 \Delta S_t$ $\hat{P}$-a.s.
\end{theorem}

The proofs of \cref{thm. exis_new} and \cref{thm. unique_new} are given in \cref{sec:app_utility}. We first establish \cref{thm. exis_new} in the one-period case ($T=1$) and then extend it to the general multi-step setting and consider the uniqueness.

\subsection{Examples and comments}
To illustrate the above results, we discuss several examples. We start with a simple example for non-uniqueness of minimal superhedging strategies.

\begin{example}[Non-Uniqueness of minimal superhedging strategies and maximizers] \label{Ex 1}
%\marginpar{\textcolor{blue}{JW: Some explanations, possibly deleted later}}
We take $\Omega=\R_+$, where $d=1$ and $T=2$ as well as $s_0=2$. Furthermore $S_t(\omega)=\omega_t$ for $t=1,2$ and 
\begin{align*}
\Pc_t(\omega)=\{P \in \mathfrak{P}(\R_+) \}, \quad t=0,1.
\end{align*}
We want to superhedge the running minimum at time 2, i.e. $\xi(\omega)=\underline{S}_2(\omega)$. Clearly $\Qc_t(\omega)=\{ Q \in \mathfrak{P}(\R_+) \ | \ \E_Q[\Delta S_{t+1}(\omega, \cdot)]=0 \}$ for all $\omega \in \Omega^t$ and $t=0,1$. Besides it is easy to see that
\begin{align*}
\sup_{Q \in \Qc} \E_{Q}[\xi]=s_0=2,
\end{align*}
so we have some degree of freedom to choose our superhedging strategy $H \in \Tr$. As it turns out we can choose any $H_1 \in [0,1]$, which gives a wealth of $2+H_1(S_1-2)$ at time 1.
%\textcolor{blue}{We need for $S_1 \ge 2$ $$ 2+H_1(S_1-2) \ge 2,$$ which is always fulfilled if $H_1 \ge 0$ and for $S_1 < 2$ $$ 2+H_1(S_1-2) \ge S_1,$$ which is fulfilled as long as $H_1 \le 1$.}
 For time 2 we have
\begin{align*}
H_2(\omega) \in  \left\{
			\begin{array}{ll}
					[0,H_1] &\text{ if }S_1(\omega)\ge 2, \\
					\left[0, \frac{2}{S_1(\omega)}+\frac{H_1}{S_1(\omega)}(S_1(\omega)-2)\right] &\text{ if } S_1(\omega)<2.
			\end{array}
				\right.
\end{align*}
%\textcolor{blue}{Again we argue for different cases. We have to find conditions such that $$ 2+H_1(S_1-2)+H_2(S_2-S_1) \ge \min(2,S_1,S_2)$$ holds. If $S_2 \ge S_1$ this gives the necessary condition $H_2 \ge 0$. If $S_2 < S_1$ we distinguish the cases (a) $2 \le S_2 <S_1$ and (b) $S_2<2\le S_1$ and (c) $S_2<S_1< 2$. In case (a) we have $$2+H_1(S_1-2)+H_2(S_2-S_1) \ge 2,$$ which yields $$ H_2 \le \inf_{S_2 \in [2,S_1]}H_1(2-S_1)/(S_2-S_1)= H_1.$$ In case (b) and (c) we have
%$$2+H_1(S_1-2)+H_2(S_2-S_1) \ge S_2,$$ which yields $$ H_2 \le \frac{S_2-2+H_1(2-S_1)}{S_2-S_1}.$$ Note that $$\left(\frac{x-a}{x-b}\right)'=\frac{a-b}{(x-b)^2}.$$ Here $a-b:=2+H_1(S_1-2)-S_1=(2-S_1)(1-H_1).$ Thus in case (b) we have $$a-b<0$$ and $$\inf_{S_2\in [0,2]}\frac{S_2-2+H_1(2-S_1)}{S_2-S_1}=H_1.$$ In case (c) we have $a-b\ge 0$ thus $$\inf_{S_2\in [0,S_1]}\frac{S_2-2+H_1(2-S_1)}{S_2-S_1}=\frac{2}{S_1}+\frac{H_1}{S_1}(S_1-2).$$}
Note also that the superhedging cost at time 1 is given by
\begin{align*}
\pi_1(\xi)(\omega)=\sup_{Q \in \Qc_1(\omega)} \E_{Q}[\xi(\omega, \cdot)]=\left\{
			\begin{array}{ll}
					2 &\text{ if }S_1(\omega)\ge 2, \\
					S_1(\omega) &\text{ if } S_1(\omega)<2.
			\end{array}
				\right. 
\end{align*}
So according to \eqref{hj} and \eqref{defminr} we can consume 
\begin{align*}
C_1(\omega)\in \left\{\begin{array}{ll}
					[0,H_1(S_1(\omega)-2)] &\text{ if }S_1(\omega)\ge 2, \\
					\left[0,(H_1(\omega)-1)(S_1(\omega)-2)\right] &\text{ if } S_1(\omega)<2
			\end{array}
				\right. 
\end{align*}
at time 1. \\
We now show that if \cref{Ass 3} is not satisfied (namely $\Pc^u$ does not fulfil \cref{Ass 3}.1.), then \cref{thm. unique_new} is not true in general. For this we specify the set $\Pc^u$ and iteratively solve the optimization problem \eqref{eq. optpro}: We set $U(2,\omega,x)=U(1,\omega,x)=U(x)$ for some bounded concave, non-decreasing and continuous function $U: \R_+ \to \R_+$ as well as $\Pc^u_1(S_1)=\{\delta_{S_1}\}$ for $S_1>2$ and $\Pc^u_1(S_1)=\{\delta_{S_1+1}\}$ for $S_1\le 2$. Note that $\Pc^u_1$ obviously violates \cref{Ass 3}.1. We obtain the following optimal one-step prices, where we use notation from \cref{Sec Mulit}: For $S_1 >2$ and $x \ge 2$ we find
\begin{align*}
U_1(S_1,x)&= \sup_{(H,c) \in \mathcal{A}_{1,x}(S_1)}\left( \E_{\delta_{ S_1}}[U(x+H(S_2-S_1)-\underline{S_2}-c)]+U(c)\right)\\
&=\sup_{(H,c) \in \mathcal{A}_{1,x}(S_1)} \left(U(x-2-c)+U(c)\right)=2U\left(\frac{x-2}{2}\right)
\end{align*}
with $c=(x-2)/2$ and some $0 \le H\le \min(\frac{x/2+1}{S_1},\frac{x/2-1}{S_1-2})$. 
%\textcolor{blue}{Indeed, note that by concavity $$U(x-2/2)=U(\frac{x-2-c+c}{2}) \ge 1/2U(x-2-c)+1/2U(c).$$ For $H$ we need to have $$\frac{x}{2}+1+H(S_2-S_{1}) \ge \min(2,S_2),$$ which is always satisfied for $S_2 \ge S_1$ and $H \ge 0$. For $S_2 < S_1$ we have Case (a): $2\le S_2 < S_1$. Then we need to have $$\frac{x}{2}+1+H(S_2-S_{1}) \ge 2$$ or $$H \le (1-x/2)/(S_2-S_1),$$ which is minimised for $S_2=2$ giving $(x/2-1)/(S_1-2)$. Case (b): $S_2 <2 <S_1$, which yields $$ H \le (S_2-x/2-1)/(S_2-S_1).$$ The RHS is minimised either for $S_2=2$ or $S_2=0$ depending on $S_1>x/2+1$ or $S_1 \le x/2+1$.   }
For $S_1 \le 2$ and $x \ge S_1$ we have
\begin{align*}
U_1(S_1,x)&= \sup_{(H,c) \in \mathcal{A}_{1,x}(S_1)} \left(\E_{\delta_{ S_1+1}}[U(x+H(S_2-S_1)-\underline{S_2}-c)]+U(c)\right)\\
&=\sup_{(H,c) \in \mathcal{A}_{1,x}(S_1)} (U(x+H-S_1-c)+U(c))\ge U(0)+U(1)
\end{align*}
with $H =x/S_1$ and $c=0$. %\textcolor{blue}{For $H$ we need to have $$x-c+H(S_2-S_{1}) \ge \min(S_1,S_2),$$ which is always satisfied for $S_2 \ge S_1$ if $H \ge 0$ (and $x-c \ge S_1)$. For $S_2 < S_1$ we have $S_2 < S_1\le 2$. This gives $$x-c+H(S_2-S_1) \ge S_2,$$ which yields $$ H \le (S_2-x+c)/(S_2-S_1).$$ The RHS is minimised for $S_2=0$ if $x-c \ge S_1 $ giving $H \le (x-c)/S_1$.}
 Setting $\Pc^u_0=\{\delta_x \ | \ x \in \R_+\}$ we obtain
\begin{align*}
U_0(2)&=\sup_{H \in \mathcal{A}_{0,2}} \inf_{P \in \Pc^u_0} \E_{P}[U_1(S_1, 2+H(S_1-2))]\\
&=\sup_{H \in \mathcal{A}_{0,2}} \inf_{P \in \Pc^u_0} \E_{P}\bigg[\mathds{1}_{\{S_1>2 \}} 2U\left(\frac{2+H(S_1-2)-2}{2} \right)\\
&+ \mathds{1}_{\{S_1 \le 2\}} U_1(2+H(S_1-2))\bigg]\\
&=2U\left(0 \right).
\end{align*}
%\textcolor{blue}{The first term is trivially minimised in $S_1$ by $S_1 \downarrow 2$, while the second term is always greater equal $U(0)+U(1)$.}
Note that by the proof of \cref{thm. unique_new} under \cref{Ass 3} there would exist $\hat{P}\in \mathcal{P}_0^u$ such that
\begin{align*}
U_0(2)=\sup_{H \in \mathcal{A}_{0,2}} \E_{\hat{P}} [U_{1}(S_1,x+H \Delta S_{1})].
\end{align*}
On the contrary, in our case there exists no $\hat{P}\in \mathcal{P}_0^u$ such that
\begin{align*}
U_0(2)=2U(0)=  \E_{\hat{P}}\bigg[\mathds{1}_{\{S_1 >2\}} 2U\left(\frac{S_1-2}{2} \right)+ \mathds{1}_{\{S_1 <2\}} U_1(S_1, 2)\bigg]
\end{align*}
as the RHS is strictly greater than $2U(0)$ for all $\hat{P}\in \mathcal{P}_0^u$: Thus \cref{thm. unique_new} does not hold.
%\textcolor{blue}{The first term is maximised for $H=1$ while the second term is decreasing in $H$, so maximised for $H=0$.}
%The set of superhedging strategies at time 2 under consideration of this additional consumption changes to
%\begin{align*}
%H_2^*(\omega) \in  \left\{
%			\begin{array}{ll}
%					\{0\} &\text{ if }S_1(\omega)\ge 2, \\
%					\left[0,1\right] &\text{ if } S_1(\omega)<2.
%			\end{array}
%				\right.
%\end{align*}
%This yields a consumption of $\Delta C_2=\pi_1(\xi)+H_2^*(S_2-S_1)-\underline{S}_2$.
\end{example}

The next example shows that we cannot expect to have uniqueness of maximizers without assuming some closedness property of $\Pc^{u}$.
\begin{example}[Non-uniqueness of maximisers for non-closed $\Pc^{u}$]\label{ex. non-unique}
Let $T=1$, $d=2$, $\Omega =\R^2$, $\Pc=\mathfrak{P}(\R_+^2)$, $S_t(\omega)=\omega_t$ and $S_0=(1,1)$. Consider $\xi=\min(S^1_1, S^2_1)$. Then $\pi(\xi)=1$ and $H_1$ is of the form
\begin{align*}
H_1= \left(
			\begin{array}{c}
					\lambda\\
					1-\lambda
			\end{array}
				\right),
\end{align*}
where $\lambda \in [0,1].$ Take 
\begin{align*}
\Pc^u=\{ P_n\}_{n=1}^{\infty} \hspace{0.5cm} \text{where } P_n=\frac{\delta_{\left\{S_1^1=n-\frac{1}{n}, \ S^2_1=n+\frac{1}{n} \right\}}}{2}+\frac{\delta_{\left\{S_1^1=0, \ S^2_1=0 \right\}}}{2}.
\end{align*}
Then clearly $\Pc^u$ is not closed. We note that for $H \in \mathcal{A}_1$
\begin{align*}
\E_{P_n}\left[U\left(1+H \Delta S_1-\xi\right)\right]&= %\E_{P_n}\left[U\left(\lambda S_1^1+ %(1-\lambda)S_1^2-\min(S_1^1,S^2_1\right)\right]\\
\frac{1}{2}U\left(\lambda\left(n-\frac{1}{n}\right)+(1-\lambda)\left(n+\frac{1}{n}\right)-\left(n-\frac{1}{n}\right)\right)+\frac{1}{2}U(0)\\
&=\frac{1}{2}U\left(\left(1-\lambda\right)\frac{2}{n} \right)+\frac{1}{2}U(0)\downarrow U(0), \ n\to \infty.
\end{align*}
Thus we conclude
\begin{align*}
\sup_{H \in \mathcal{A}_1} \inf_{P \in \Pc^u} \E_{P}[U(1+ H \Delta S_1-\xi)]=U(0),
\end{align*}
in particular
\begin{align*}
H \mapsto \inf_{P \in \Pc^u} \E_{P}[U(1+ H \Delta S_1-\xi)]=U(0)
\end{align*}
is constant and thus the maximizer is not unique.
\end{example}
Finally, we illustrate that even with a compact $\Pc^{u}$ we can not strengthen the sense in which the optimisers are unique in \cref{thm. unique_new}.
\begin{example}[On uniqueness property of maximisers]
We consider a one-step version of \cref{Ex 1}: $T=1$, $d=1$, $\Omega=\R_+$, $S_t(\omega)=\omega_t$, $s_0=2$, $\xi(S)=\underline{S_1}$, $\Pc=\mathfrak{P}(\R_+)$. We have $\pi(\xi)=2$. We also set $\Pc^u=\{\delta_2 \}$, where $\delta_2$ is defined by $$\delta_2( S_t=2 \ \text{for all } t=0, 1)=1.$$  Furthermore let $U(\cdot)=U(1, \cdot, \cdot)$ such that the conditions of \cref{thm. unique_new} are satisfied. The optimisers are then non-unique in the sense that \eqref{eq. optpro} is equal to $U(0)$ and is attained for every $H \in [0,1]$ but are unique in the sense of \cref{thm. unique_new} since $H\Delta S_1=0$ $\delta_2$-a.s.\ for all $H \in \R$.
\end{example}

\newpage
\begin{appendices}
\appendixpage
We now provide the proofs of \cref{prop:pathwiseex}, \cref{eu}, \cref{thm. exis_new} and of \cref{thm. unique_new}. These proofs require a number of technical lemmata which are established alongside the main proofs.
\section{Proof of \cref{prop:pathwiseex}}\label{appendix:A}

\begin{proof}%[of Proposition \cref{prop:pathwiseex}]
Fix $\omega \in \Omega^{T-1}$ and $\epsilon>0$. Recall that $\xi$ is continuous and $\{f_t\}_{0\le t\le T-1}$ are compact-valued. Note that the set $$B:=\{(\tilde{\omega},\tilde{v})\in \Omega^{T-1}\times \R^d \ | \  \text{dist}((\omega,f_{T-1}(\omega)),(\tilde{\omega},\tilde{v}))\le 1\}$$ is compact, thus $\xi$ is uniformly continuous on $B$, i.e. there exists $\delta\in (0,1)$ such that $|\xi(\omega, v)- \xi(\tilde{\omega}, \tilde{v})|\le \epsilon/3$ for $|(\omega,v)-(\tilde{\omega},\tilde{v})|\le \delta$ for $v \in f_{T-1}(\omega)$, $(\tilde{\omega}, \tilde{v}) \in B$. This implies $\sup_{\{\tilde{\omega}| \ |\omega-\tilde{\omega}|\le 1\}} \pi_{T-1}(\xi)(\tilde{\omega})<\infty$ and that for all $\tilde{\omega}\in \Omega^{T-1}$ with $|\omega-\tilde{\omega}|\le 1$ there exists $H_{T}(\tilde{\omega}) \in \R^d$ such that
\begin{align}\label{eq. superhedgingl}
\epsilon/3+\pi_{T-1}(\xi)(\tilde{\omega})+H_{T}(\tilde{\omega}) \Delta S_{T}(\tilde{\omega},\cdot) \ge \xi(\tilde{\omega}, \cdot) \quad \text{on } f_{T-1}(\tilde{\omega})
\end{align}
or equivalently the inequality \eqref{eq. superhedgingl} holds $ \Pc_{T-1}(\tilde{\omega})\text{-q.s}$.\\
Note that by the uniform continuity of the correspondence $f_{T-1}$  for any $\tilde{\omega}$ close to $\omega$ and for any   $v \in f_{T-1}(\omega)$ there exists $\tilde{v} \in f_{T-1}(\tilde{\omega})$ which is close to $v$, thus $|(\omega,v)-(\tilde{\omega},\tilde{v})|$ is small. Furthermore we show below that $H_T(\tilde{\omega})$ can be chosen bounded uniformly in $\tilde{\omega}$ for all $\tilde{\omega}$ close to $\omega$. Thus, for some $\delta_1$ determined below, $|\omega-\tilde{\omega}|\le \delta_1$ implies
\begin{align}\label{eq. epsilon}
\epsilon+\pi_{T-1}(\xi)(\tilde{\omega})+H_{T}(\tilde{\omega}) \Delta S_{T}(\omega,v) &\ge 
\epsilon+\pi_{T-1}(\xi)(\tilde{\omega})+H_{T}(\tilde{\omega}) \Delta S_{T}(\tilde{\omega},\tilde{v})-\epsilon/3\\
&\ge \epsilon/3+\xi(\tilde{\omega},\tilde{v}) \ge \xi(\omega,v),\nonumber
\end{align}
and thus $\pi_{T-1}(\xi)(\omega) \le \pi_{T-1}(\xi)(\tilde{\omega})+\epsilon$. Exchanging the roles of $\omega$ and $\tilde{\omega}$ concludes the proof of continuity of $\omega \mapsto \pi_{T-1}(\omega)$. \\
We now argue that there exists $\delta_0>0$ and $C>0$ such that $|H_{T}(\tilde{\omega})|<C$ for all $\tilde{\omega}\in \Omega^{T-1}$ with  $|\omega-\tilde{\omega}|\le \delta_0$ and $H_T(\tilde{\omega})\in \text{lin}(f_{T-1}(\tilde{\omega})-S_{T-1}(\tilde{\omega}))$. Assume towards a contradiction this is not the case, i.e. there exists a sequence $(\tilde{\omega}^N)_{N\in \N}$ with $|\omega-\tilde{\omega}^N|\le 1/N$, $H_T(\tilde{\omega}^N)\in \text{lin}(f_{T-1}(\tilde{\omega}^N)-S_{T-1}(\tilde{\omega}^N))$  for all $N\in \N$ and $\lim_{N\to \infty}|H_T(\tilde{\omega}^N)|=\infty$. After passing to a subsequence (without relabelling) $\tilde{H}^N:=H_T(\tilde{\omega}^N)/|H_T(\tilde{\omega}^N)|\to \tilde{H}$ with $|\tilde{H}|=1$. Note that as $f_{T-1}(\tilde{\omega}^N)$ converges in Hausdorff distance to $f_{T-1}(\omega)$ and as $f_{T-1}(\omega)$ is compact, it follows by the same arguments as above that $\sup_{f_{T-1}(\tilde{\omega}^N)}\xi(\tilde{\omega}^N,\cdot)$ and $\pi_{T-1}(\xi)(\tilde{\omega}^N)$ are bounded uniformly in $N\in \N$. Thus dividing \eqref{eq. superhedgingl} by $|H_T(\tilde{\omega}^N)|$ and taking limits we get
\begin{align*}
\tilde{H}\Delta S_{T}(\omega,\cdot)\ge 0 \quad\text{on }f_{T-1}(\omega).
\end{align*}
By NA$(\Pc_{T-1}(\omega))$ this yields $\tilde{H}\Delta S_T(\omega, \cdot)=0$ on $f_{T-1}(\omega)$. As $\tilde{H}\in \text{span}(f_{T-1}(\omega)-S_{T-1}(\omega))$, $\tilde{H}=0$ follows, a contradiction.\\
Now we choose $\delta_1\le \delta_0$ such that for $|\omega- \tilde{\omega}|\le \delta_1$ we have $$d_H((\omega,f_{T-1}(\omega)), (\tilde{\omega},f_{T-1}(\tilde{\omega})) \le \min(\delta, \epsilon/(3C))$$ and see that \eqref{eq. epsilon} holds. The proof of continuity of $\omega \mapsto \pi_t(\xi)(\omega)$ for $1\le t\le T-2$ follows by backward induction using dynamic programming principle and the same arguments as above.
Lastly, as for any $P \in \mathfrak{P}(\R^d)$ such that $\text{supp}(P)=f_{t-1}(\omega)$ 
\begin{align*}
\pi_{t-1}(\xi)(\omega)+H_{t}(\omega) \Delta S_{t}(\omega, \cdot) \ge \pi_t(\xi)(\omega,\cdot) \quad P\text{-a.s.}
\end{align*}
implies 
\begin{align*}
\pi_{t-1}(\xi)(\omega)+H_{t}(\omega) \Delta S_{t}(\omega, \cdot) \ge \pi_t(\xi)(\omega,\cdot) \quad \text{ on }f_{t-1}(\omega),
\end{align*}
the claim follows.
\end{proof}

\begin{remark}\label{rem:bounded}
Note that the proof of boundedness of $H_{T}(\tilde{\omega})$ above does not require that $f_{T-1}(\tilde{\omega})$ is compact-valued.
\end{remark}

\section{Proof of \cref{eu}}\label{appendix:laurence}
\begin{lemma}
\label{olala}
Let NA$(\Pc)$ hold. Assume that $\xi$ is upper semianalytic. Furthermore let \mbox{$\sup_{Q\in\Qc} \EE_Q [\xi^-  ]< \infty.$} 
Then $\Ec^{t}(\xi)$ is upper semianalytic and $\Ec^{t}(\xi^-)$ is lower semianalytic for all $0\le t\le T-1$. Furthermore
\begin{align*}
\sup_{Q \in \Qc}\EE_{Q}[\Ec^{t}(\xi^-)]<\infty
\end{align*}
and the analytic set $\O^t_{\xi}:=\{\Ec^{t}(\xi^-)<\infty\}$ is of full $\Pc$-measure. Let
\begin{align}
\label{pasmes}
\hat{\O}^t_{\xi}:=\{\omega \in \Omega^t\ | \ \Ec^{t+1}(\xi)(\omega,\cdot)>-\infty, \; \Pc_{t}(\omega)\mbox{-q.s.}\}.
\end{align}
Then ${\O}^t_{\xi}  \subset \hat{\O}^t_{\xi}$, in particular $\hat{\O}^t_{\xi}$ is a $\Pc$-full measure set.
\end{lemma}

\begin{proof}
Using  \cite[Lemma 4.10]{BN} recursively, $\Ec^{t}(\xi)$ is upper semianalytic and $\Ec^{t}(\xi^-)$ is lower semianalytic for all $0 \leq t \leq T$. %(see \cite[Proposition 7.38 p165, Proposition 7.48 p180, Proposition 7.47 p179]{bs} since $\Qc_{t}$ has a non-empty analytic graph (see \cite[Lemma 4.8]{BN}).

As ${\Omega}_{\xi}^{t} = \{ \Ec^{t}(\xi^-) < \infty \}= \bigcup_{n \ge 1} \{ \Ec^{t}(\xi^-) \le n\}$, $\O^t_{\xi}$ is an analytic set.
We now prove by induction that ${\Omega}_{\xi}^{t}$ is a $\Pc$-full measure set and that $\sup_{Q \in \Qc}\EE_{Q}[\Ec^{t}(\xi^-)]<\infty$. For $t=T$, $\sup_{Q \in \Qc}\EE_{Q} [\xi^-]<\infty$ by assumption. If there exists some  $P\in \Pc$ such that $P({\Omega}_{\xi}^{T})<1$ then   $\sup_{Q \in \Qc}\EE_{Q}[\Ec^{T}(\xi^-)] =\infty$, as $\Pc$ and $\Qc$ have the same polar sets (see \cite[First Fundamental Theorem, p. 828]{BN}).

Assume for some $ t\leq T-1$ that ${\Omega}_{\xi}^{t+1}$ is a $\Pc$-full measure set and that\\
$\sup_{Q \in \Qc}\EE_{Q}[\Ec^{t+1}(\xi^-)]<\infty.$
Fix $\epsilon>0$. From  \cite[Proposition 7.50 p184]{bs} (recall that $\Qc_{t}$ has an analytic graph), there exists an $\Fut$-measurable function $Q_{\epsilon}: \Omega^t \to \mathfrak{P}(\Omega)$, such that  $Q_{\epsilon}(\omega) \in \Qc_{t}(\omega)$ for all $\omega \in \Omega^{t}$ and
\begin{align}
\label{eqjt}
\E_{Q_{\epsilon}} [\Ec^{t+1}(\xi^-(\omega,\cdot)] \ge \begin{cases} \Ec^t (\xi^-)(\omega)- \varepsilon &\mbox{if $\o \in {\Omega}_{\xi}^{t}$},\\
\frac{1}{\varepsilon} \; &\mbox{otherwise}.
\end{cases}
\end{align}
Assume that $\Omega_{\xi}^{t}$ is not a $\Pc$-full measure set. Then there exists some $P \in \Pc$ such that $P({\Omega}_{\xi}^{t} )<1$. As $\Pc$  and $\Qc$ have the same polar sets, we have that $Q({\Omega}_{\xi}^{t} )<1$ for some $Q \in \Qc$. We denote by $Q|_{\Fc_t^{\cal{U}}}$ the restriction of $Q$ to ${\Fc_t^{\cal{U}}}$ and set $Q^{*}:=Q|_{\Fc_t^{\cal{U}}}\otimes Q_{\varepsilon} $. Then $Q^{*}\in \Qc|_{\Fc_{t+1}^{\cal{U}}}$ (see \eqref{Mstar}) and
we have that
\begin{align*}
\sup_{Q \in \Qc} \EE_Q [\Ec^{t+1}(\xi^-)] \ge \EE_{Q^{*}} [\Ec^{t+1}(\xi^-)] \ge \frac{1}{\epsilon}  (1-Q^{*}({\Omega}_{\xi}^{t} )) - {\epsilon}Q^{*}({\Omega}_{\xi}^{t} ) .
\end{align*}
As the previous inequality holds for all $\epsilon>0$, letting $\epsilon$ go to $0$ we  obtain that $$\sup_{Q \in \Qc}\EE_{Q}[\Ec^{t+1}(\xi^-)]=\infty,$$ a contradiction. Thus $\Omega_{\xi}^{t}$ is a $\Pc$-full measure set.\\
Now,  for all  $Q \in \Qc$, we set $Q^*=Q|_{\Fc_t^{\cal{U}}}\otimes Q_{\epsilon} \in \Qc|_{\Fc_{t+1}^{\cal{U}}}$ (see \eqref{Mstar}).  Then, using  \eqref{eqjt} we see that
\begin{align*}
\EE_{Q} [\Ec^{t}(\xi^-)] -\varepsilon = \EE_{Q} [\mathds{1}_{\Omega^{t}_{\xi}}  \Ec^{t}(\xi^-)]  -\epsilon \le \EE_{Q^*}[ \Ec^{t+1}(\xi^-)] \le
 \sup_{Q \in \Qc}\EE_{Q}[\Ec^{t+1}(\xi^-)].
\end{align*}
Again, as this is true for all $\epsilon>0$ and all $Q \in \Qc$ we obtain that $ \sup_{Q \in \Qc}\EE_{Q}[\Ec^{t}(\xi^-)] \le \sup_{Q \in \Qc}\EE_{Q}[\Ec^{t+1}(\xi^-)]<\infty.$\\
Let $0\le t \le T-1$ and $\omega \in \Omega^t_{\xi}$. Then for all $Q \in \Qc_{t}(\omega)$, $\EE_Q[\Ec^{t+1}(\xi^-)(\omega, \cdot)]<\infty$, which implies that $\Ec^{t+1}(\xi^-)(\omega, \cdot)<\infty$ $Q$-a.s. and thus $\Ec^{t+1}(\xi^-)(\omega, \cdot)<\infty$  $\Pc_{t}(\omega)$-q.s. Assume for a moment that we have proved $\Ec^{t+1}(\xi) \geq -\Ec^{t+1}(\xi^-)$. Then  $-\Ec^{t+1}(\xi)(\omega, \cdot)<\infty$  $\Pc_{t}(\omega)$-q.s. and
$\omega \in \hat{\Omega}^t_{\xi}$. Thus ${\Omega}^t_{\xi}  \subseteq \hat{\Omega}^t_{\xi}$ and $\hat{\Omega}^t_{\xi}$ is a $\Pc$-full measure set.\\
 %but we have not proved any measurability for $\hat{\O}^t_h$ (${\O}^t \setminus\hat{\O}^t_h  \subset {\O}^t \setminus {\O}^t_h$ and ${\O}^t \setminus {\O}^t_h$ is co-ana thus univ-mes and of measure $0$ for all $Q\in  \mathcal{P}_{t+1}(\omega^{t})$ which means $\mathcal{P}_{t+1}(\omega^{t})$-polar).
Let $0\le t \le T-1$. We now prove that $\Ec^{t+1}(\xi) \geq -\Ec^{t+1}(\xi^-)$ by backward induction. The claim is clearly true for $t=T-1$. Assume that it is true for some $1\le t+1 \le T$. Then for $\omega \in \Omega^t$ we find
\begin{align*}
\Ec^{t}(\xi)(\omega) & =  \sup_{Q \in \Qc_{t}(\omega)}\EE_Q[\Ec^{t+1}(\xi)(\omega, \cdot)] \ge  \sup_{Q \in \Qc_{t}(\omega)}\EE_Q[-\Ec^{t+1}(\xi^-)(\omega, \cdot)] \\
&  \ge \inf_{Q \in \Qc_{t}(\omega)}\EE_Q[-\Ec^{t+1}(\xi^-)(\omega, \cdot)]= -\Ec^{t}(\xi^-)(\omega).
\end{align*}
This concludes the proof.
\end{proof}

\begin{remark}\label{rem. laurence}
Recall the set $\Omega_{\text{NA}}^t=\{\omega \in \Omega^t \ | \ \text{NA}(\Pc_t(\omega))\text{ holds} \}$,
which is universally measurable and of $\Pc$-full measure (see \cite[Lemma 4.6, p.842]{BN}). Let $\omega \in \Omega^t_{\text{NA}}$. From \cite[Lemma 4.1]{BN}, we know that
$\Ec^{t}(\xi)(\omega)=-\infty$ implies that $\{\Ec^{t+1}(\xi)(\omega,\cdot)=-\infty\}$ is not $\Pc_{t}(\omega)$-polar i.e. $\omega \notin \hat{\Omega}^t_{\xi}$. Thus
$${\O}^t_{\xi} \cap \O^t_{\text{NA}} \subseteq \hat{\O}^t_{\xi} \cap \O^t_{\text{NA}} \subseteq \{\o \in \Omega_{\text{NA}}^t\ | \ \Ec^{t}(\xi)(\omega)>-\infty\}.$$
\end{remark}

\begin{lemma} \label{lem::usa}
If $\xi: \Omega^T \to \R$ is upper semianalytic, then $\pi_t(\xi)$ is upper semianalytic for all $0 \le t \le T-1$. 
\end{lemma}
\begin{proof}
We proceed by induction. As $\pi_T(\xi)=\xi$ the claim is true for $t=T$. Assume now the $\pi_{t+1}(\xi)$ is upper semianalytic for some $t\in \{0, \dots, T-1\}$. We show that the claim is true for $t$. Indeed for all $a \in \R$
\begin{align*}
&\{\omega \in \Omega^{t} \ | \ \pi_t(\xi) < a \}\\
=\  &\{\omega \in \Omega^t  \ | \ \exists H \in \R^d, \ \epsilon>0 \text{ s. t. }\forall P \in \mathcal{P}_t(\omega) \ P(a-\epsilon+H \Delta S_{t+1}(\omega, \cdot) \ge \pi_{t+1}(\xi)(\omega, \cdot))=1 \}\\
=\ &\{\omega \in \Omega^t \ | \ \sup_{\epsilon \in \QQ_+}\sup_{H \in \QQ^d} \inf_{P \in \mathcal{P}_t(\omega)} P(a-\epsilon+H \Delta S_{t+1}(\omega, \cdot) \ge \pi_{t+1}(\xi)(\omega, \cdot))\ge 1 \}
\end{align*}
As the function $(\omega, P,H,\epsilon) \mapsto \E_{P}\left[\mathds{1}_{\{a-\epsilon+H \Delta S_{t+1}(\omega, \cdot) \ge \pi_{t+1}(\xi)(\omega, \cdot)\}}\right]$ is lower semianalytic, the same holds true for $\omega \mapsto \sup_{\epsilon \in \QQ_+}\sup_{H \in \QQ^d} \inf_{P \in \mathcal{P}_t(\omega)} \E_{P}\left[\mathds{1}_{\{a-\epsilon+H \Delta S_{t+1}(\omega, \cdot) \ge \pi_{t+1}(\xi)(\omega, \cdot)\}}\right]$ (see \cite[Lemma 7.30, p.177, Prop. 7.47, p.180]{bs}), thus the set above is coanalytic. To complete the proof, we argue why 
\begin{align*}
&\{\omega \in \Omega^t \ | \ \exists H \in \R^d, \ \epsilon>0 \text{ such that }a-\epsilon+H\Delta S_{t+1}(\omega, \cdot) \ge \pi_{t+1}(\omega, \cdot) \ \Pc_t(\omega)\text{-q.s.}\} \\
\subseteq\ &\{\omega \in \Omega^t \ | \ \exists H \in \QQ^d,\ \epsilon\in \QQ_+ \text{ such that }a-\epsilon+H\Delta S_{t+1}(\omega, \cdot) \ge \pi_{t+1}(\omega, \cdot) \ \Pc_t(\omega)\text{-q.s.}\}:
\end{align*}
Fix $\omega \in \Omega^t$, $\tilde{H}\in \R^d$, $\epsilon>0$ such that $a-\epsilon+\tilde{H}\Delta S_{t+1}(\omega, \cdot) \ge \pi_{t+1}(\omega, \cdot) \ \Pc_t(\omega)\text{-q.s.}$ Take $\tilde{\epsilon}\in \QQ_+$ such that $0<\tilde{\epsilon}<\epsilon/2$ and $H \in [0, \infty)^d$ such that 
\begin{align*}
H^1+ \dots +H^d \le \frac{\epsilon/2}{\max_{1\le i \le d}S^i_t(\omega)}.
\end{align*}
%Note that there exists $\tilde{H}\in \R^d$ such that
%\begin{align*}
%\pi_t(\xi)(\omega)+\tilde{H}\Delta S_{t+1}(\omega, \omega') \ge \pi_{t+1}(\xi)(\omega, \omega')
%\end{align*}
%for $\Pc_t(\omega)$-q.e. $\omega'\in \Omega$. 
It follows that for $\Pc_t(\omega)$-q.e. $\omega' \in \Omega$
\begin{align*}
a-\tilde{\epsilon}+(H+\tilde{H})\Delta S_{t+1}(\omega, \omega') &\ge a -\epsilon/2+ \tilde{H} \Delta S_{t+1}(\omega, \omega')
+H\Delta S_{t+1}(\omega, \omega')\\
&\ge \pi_{t+1}(\xi)(\omega, \omega')+\epsilon/2-H S_t(\omega) \\
&\ge \pi_{t+1}(\xi)(\omega, \omega').
\end{align*}
In particular the above inequality is valid for some $H$ such that $\tilde{H}+H \in \QQ^d$.
\end{proof}

%\begin{lemma}
%The set $\{\omega \in \Omega^t \ | \ \pi_t(\xi) \ge \mathcal{E}^{t}(\xi)\}$ is universally measurable and there exists an $\Fut$-measurable selector
%\end{lemma}
%\begin{proof}
%Similarly 
%\begin{align*}
%\{\omega \in \Omega^t \ | \ \pi_t(\xi)>\mathcal{E}^t(\xi)\}
%&= \left\{ \omega \in \Omega^t \ \bigg| \ \sup_{H \in \QQ^d} \inf_{P \in \mathcal{P}_t(\omega)} P(\mathcal{E}^t(\xi)(\omega)+H \Delta S_{t+1}(\omega, \cdot) < \pi_{t+1}(\xi)(\omega, \cdot))>0 \right\}\\
%&= \bigcap_{H \in \QQ^d}\text{proj}_{\Omega^t}\left(\left\{ (\omega,P) \ \ \bigg| \ P(\mathcal{E}^t(\xi)(\omega)+H \Delta S_{t+1}(\omega, \cdot) < \pi_{t+1}(\xi)(\omega, \cdot))> 0 \right\}\right)
%\in \Fut.
%\end{align*}
%and for every $H \in \QQ^d$ there exists an $\Fut$-measurable selector $P^H: \Omega_t \to \mathfrak{P}(\Omega)$. Assume there exists $P \in \mathcal{P}$ such that $P(\pi_t(\xi)>\mathcal{E}^t(\xi))>0$. Take $(q)_{n \in \N}$ an enumeration of $\QQ^d$ and set $\hat{P}:= \sum_{q \in \QQ^d} 2^{-q}P^{q}.$ Now define $P^*=P|_{\Fut}\otimes \hat{P}$. Then $P^*(\forall H \in \QQ^d \mathcal{E}^t(\xi)+H\Delta S_{T-1}< \mathcal{E}^{t+1}(\xi))>0,$ contradicting $\pi^{t+1}=\mathcal{E}^{t+1}$.
%\end{proof}

\begin{proof}[of \cref{eu}]

Let
\begin{align*}
\Omega_{\text{NA}, \xi}:= \{\omega \in \Omega^T \ | \ \omega \in \Omega_{\text{NA}}^t \cap \Omega_{\xi}^t \text{ for all }0\le t \le T-1 \},
\end{align*} 
where the definition of $\Omega^t_{\xi}$ is given in \cref{olala} and the definiton of $\Omega_{\text{NA}}^t$ in \cref{rem. laurence}. Then by \cref{olala} and \cite[Lemma 4.6, p. 842]{BN} $\Omega_{\text{NA}, \xi}$ is universally measurable and of $\Pc$-full measure. Let $\omega \in \Omega_{\text{NA}, \xi}$. By  \cite[Lemma 4.10]{BN}, there exists a universally measurable function $\hat{H}_{t+1}$ such that
\begin{align}
\label{fantarec}
\mathcal{E}^{t}(\xi) (\omega) + \hat{H}_{t+1} (\omega) \Delta S_{t+1}(\omega, \cdot) \ge \mathcal{E}^{t+1}(\xi) (\omega,\cdot)  \quad  \Pc_{t}(\omega)\mbox{-q.s.}
\end{align}
%and thus by definiton of $\pi_t(\xi)$ we have $\pi_t(\xi)(\omega) \le \mathcal{E}^t(\xi)(\omega)$. As the opposite inequality follows from the definition of $\pi_t(\xi)$ and the fact that $\Pc$ and $\Qc$ have the same polar sets, the claim follows for all $\omega \in \Omega_{\text{NA},\xi}$.\\}
To see that 
\begin{align}
\label{beindidonc}
\Xt = \Ec^t(\xi)  \quad \Pc\text{-q.s.}
\end{align}
for $0\le t \le T$ we argue by backwards induction. Indeed the claim is true by definition for $t=T$. Now we assume that the claim is true for $t+1 \in \{1, \dots, T\}$. By \cite[eq. (4.8) in Lemma 4.8, p.843]{BN} the correspondence 
\begin{align*}
\mathcal{H}_t(\omega)=\{  (Q, P) \in \mathfrak{P}(\Omega) &\times \mathfrak{P}(\Omega) \ | \ \E_{Q}[\Delta S_{t+1}(\omega, \cdot)]=0, \ P \in \mathcal{P}_t(\omega), \ Q \ll P \}
\end{align*}
has analytic graph. By \cite[Prop. 7.47, p. 179, Prop. 7.48, p. 180, Prop. 7.50, p.184]{bs} $(\omega,Q, P) \mapsto \E_{Q}[\mathcal{E}_{t+1}(\xi)(\omega,\cdot)]$ and $(\omega,Q, P) \mapsto \E_{Q}[\pi_{t+1}(\xi)(\omega, \cdot)]$ are upper seminanalytic functions and there exists sequences $(\hat{P}_n,\hat{Q}_n)_{n \in \N}$ and $(\bar{P}_n,\bar{Q}_n)_{n \in \N}$ of $\Fut$-measurable selectors of $\mathcal{H}_t$ such that 
\begin{align*}
\lim_{n \to \infty} \E_{\hat{Q}_n(\omega)}[\mathcal{E}^{t+1}(\xi)(\omega, \cdot)] &=\sup_{(Q,P)\in \mathcal{H}_t(\omega)}\E_{Q}[\mathcal{E}^{t+1}(\xi)(\omega, \cdot)]= \mathcal{E}^t(\xi)(\omega),\\
\lim_{n \to \infty} \E_{\bar{Q}_n(\omega)}[\pi_{t+1}(\xi)(\omega,\cdot)]&=\sup_{(Q,P)\in \mathcal{H}_t(\omega)}\E_{Q}[\pi_{t+1}(\xi)(\omega,\cdot)]= \mathcal{E}_t(\pi_{t+1}(\xi))(\omega).
\end{align*}
Define $P_n(\omega)=(\hat{P}_n(\omega)+\bar{P}_n(\omega))/2\in \Pc_t(\omega)$ and $\tilde{P}_t(\omega)= \sum_{n =1}^{\infty}2^{-n}P_n(\omega)$. Then $\tilde{P}_t(\omega) \in \mathfrak{P}(\Omega)$ for all $\omega \in \Omega^t$, $\omega \mapsto \tilde{P}_t(\omega)$ is $\Fut$-measurable and $\hat{P}_n(\omega), \bar{P}_n(\omega), P_n(\omega)$ are absolutely continuous with respect to $\tilde{P}_t(\omega)$. Furthermore for $\omega \in \Omega_{\text{NA}}^t$ 
\begin{align*}
\E_{\hat{Q}_n(\omega)}[\mathcal{E}^{t+1}(\xi)(\omega, \cdot)]&\le \sup_{Q \ll \tilde{P}_t(\omega), \ \E_{Q}[\Delta S_{t+1}(\omega, \cdot)]=0} \E_{Q}[\mathcal{E}^{t+1}(\xi)(\omega, \cdot)] \\&\le \inf\{x \in \R \ | \ \exists H\in \R^d \text{ such that }x+H\Delta S_{t+1}(\omega, \cdot) \ge  \mathcal{E}^{t+1}(\xi)(\omega, \cdot) \ \tilde{P}_t(\omega)\text{-a.s.}\}\\
&\le \pi_{t}(\mathcal{E}^{t+1}(\xi))(\omega)=\mathcal{E}_{t}(\mathcal{E}^{t+1}(\xi))(\omega)=\mathcal{E}^t(\xi)(\omega),
\end{align*}
where the third inequality follows from the fact that $P_n(\omega) \in \Pc_t(\omega)$ for $n \in \N$ and the first equality follows from \cite[Theorem 3.4]{BN} as $\omega \in \Omega_{\text{NA}}^t$. Letting $n \to \infty$ we conclude
\begin{align*}
\sup_{Q \ll \tilde{P}_t(\omega), \ \E_{Q}[\Delta S_{t+1}(\omega, \cdot)]=0} \E_{Q}[\mathcal{E}^{t+1}(\xi)(\omega, \cdot)]&= \mathcal{E}^{t}(\xi)(\omega),\\
\sup_{Q \ll \tilde{P}_t(\omega), \ \E_{Q}[\Delta S_{t+1}(\omega, \cdot)]=0} \E_{Q}[\pi_{t+1}(\xi)(\omega,\cdot)]&= \mathcal{E}_t(\pi_{t+1}(\xi)(\omega,\cdot)),
\end{align*}
Fix now $P \in \mathcal{P}$ and define $\tilde{P}=P|_{\Fut}\otimes \tilde{P}_t$. Then as $P_n(\omega)\in \mathcal{P}_t(\omega)$ the induction assumption implies that $\mathcal{E}^{t+1}(\xi)=\pi_{t+1}(\xi)$ holds $\tilde{P}$-a.s. and thus for $\tilde{P}$-a.e. $\omega \in \Omega^t$ we have
\begin{align*}
\mathcal{E}^{t}(\xi)(\omega)&=\sup_{Q \ll \tilde{P}_t(\omega), \ \E_{Q}[\Delta S_{t+1}(\omega, \cdot)]=0} \E_{Q}[\mathcal{E}^{t+1}(\xi)(\omega, \cdot)]\\
&=\sup_{Q \ll \tilde{P}_t(\omega), \ \E_{Q}[\Delta S_{t+1}(\omega, \cdot)]=0} \E_{Q}[\pi_{t+1}(\xi)(\omega, \cdot)]
\\&=\mathcal{E}_t(\pi_{t+1}(\xi))(\omega)=\pi_t(\xi)(\omega),
\end{align*}
where the last equality again follows from \cite[Theorem 3.4]{BN} if $\omega \in \Omega_{\text{NA}}^t$. This concludes the proof of \eqref{beindidonc}.\\
%\textcolor{blue}{To see that 
%\begin{align}
%\label{beindidonc}
%\Xt = \Ec^t(\xi)  \quad \Pc\mbox{-q.s.}
%\end{align}
%for $t=0, \dots, T$ we refer to \cite[Proposition 5.3]{jw}, where it is shown that
%\begin{align*}
%\Xt = \pi^{\hat{P}_t}(g)=\Ec^t(\xi) \quad \Pc\mbox{-q.s.}
%\end{align*}
%for some $\hat{P} = \hat{P}_0 \otimes \hat{P}_1 \otimes \cdots \otimes \hat{P}_{T-1}$, where the universally measurable kernels $\hat{P}_t(\omega)$ are elements of  $\overline{\mathfrak{P}}_t(\omega)$ $\Pc$-q.s. for all $t=0, \dots, T-1$. The claim thus follows by backwards induction.
%}
Let $(x,H,C) \in \Sp(\xi)$. Now we show that
\begin{align}
\label{eeu2} V_t^{x,H,C} & \geq  \Xt  \quad \Pc\mbox{-q.s.}
\end{align}
This is clearly true at $t=T$. Fix some $1\le t\le T$ and assume that \eqref{eeu2} holds true for $t$. Then
\begin{align*}
V_{t-1}^{x,H,C} + H_t \Delta S_t \geq V_{t}^{x,H,C} \ge \Xt\quad\Pc\mbox{-q.s.}
\end{align*}
Noting that $V_{t-1}^{x,H,C}$ is $\mathcal{F}^{\mathcal{U}}_{t-1}$-measurable and $\pi_t(\xi)$ is upper seminanalytic and using the same reasoning as in \cite[proof of Lemma 4.10, pp.846-848]{BN} we conclude that for $\omega\in \Omega^{t-1}$ in a $\Pc$ full-measure set
\begin{align}\label{eq. bn_lem4.10}
V_{t-1}^{x,H,C} (\omega) + H_t (\omega)\Delta S_t(\omega, \cdot) \geq \Xt (\omega, \cdot)\quad \Pc_{t-1}(\o)\text{-q.s.}
\end{align}
Thus $V_{t-1}^{x,H,C} (\omega) \geq \pi_{t-1}(\xi)(\omega)$ by \eqref{defminr} and \eqref{eeu2} is proved for $t-1$. 
Next we define the consumption process $\hat{C}$. Let $P=P_{0}\otimes P_{1} \otimes \dots \otimes P_{T-1} \in  \mathcal{P}$, where $P_{t} \in \mathcal{P}_{t}(\omega)$ for all $0\le t \le T-1$. Then using \cref{fantarec} and Fubini's Theorem (recall \cite[Proposition 7.45 p175]{bs}), we get that
\begin{align}
\label{fantarectot}
\Ec^{t-1}(\xi)  + \hat{H}_{t}  \Delta S_{t} \ge \Ec^{t}(\xi)  \quad  \Pc\mbox{-q.s.}
\end{align}
for a universally measurable function $\hat{H}_t: \O^t \to \R^d$. Using \eqref{fantarectot} recursively,
\begin{align}
\label{fantacool}
\Ec^0(\xi) + \sum_{u=1}^t \hat{H}_u \Delta S_u \ge \Ec^t(\xi)\quad  \Pc\mbox{-q.s.}
\end{align}
follows.
%let $\e>0$. We will use \eqref{beindidonc} and
% \cite[Lemma 4.10]{BN}.  By \eqref{defminr} for all $t=1,\ldots,T$, $\o^{t-1} \in \Omega^{t-1}$ there exists $\phi_{t-1}(\o^{t-1}) \in \R^d$  such that $\phi_{t-1}$ is universally measurable and
% $\bar{X}_{t-1}(\o^{t-1}) + \frac{\e}{T}+ \phi_{t-1}(\o^{t-1}) \Delta S_t(\o^{t-1}, \cdot) \geq \bar{X}_t(\o^{t-1}, \cdot)$ $\Qc_t(\o^{t-1})$-q.s.
%This implies using Fubini Theorem that
%$$\bar{X}_{t-1} + \frac{\e}{T} + \phi_{t-1}\Delta S_t \geq \bar{X}_t \;  \Qc^t\mbox{-q.s.},$$
%Summing over all $t=1,\ldots,T$
%$$\bar{X}_0 + \e + \sum_{t=1}^T \Phi_{t-1} \Delta S_t \geq h \;  \Qc^T\mbox{-q.s.},$$
%and therefore $\bar{X}_0 + \e \geq \pi(h)$ and passing to the limit  $\bar{X}_0 \geq \pi(h)$.
%ATTENTION POUR FAIRE CELA ON A UTILISE $\Gamma_t h =\Ec^t h$ qui utilise
%$\Ec^t h_0 = \pi(h)$ de BN pour eviter BN on peut faire de la selection mesurable direct et trouver $\phi_{t-1}$ is universally %measurable  alors $\bar{X}_0 = \pi(h)$ direct.
Now we set $\hat{C}_t=\Ec^0(\xi) + \sum_{u=1}^t \hat{H}_u \Delta S_u - \Ec^t(\xi)$. Then
$\hat{C}_{t}(\omega, \cdot) -\hat{C}_{t-1}(\omega)=\Ec^{t-1}(\xi)(\omega) - \Ec^{t} (\xi)(\omega, \cdot) + \hat{H}_{t}(\omega) \Delta S_{t}(\omega, \cdot) \geq 0$ $\Pc_{t-1}(\omega)$-q.s. and using again Fubini's Theorem $\hat{C}_{t} -\hat{C}_{t-1}\geq 0$ $\Pc\mbox{-q.s.}$ Thus $\hat{C}=(\hat{C}_t)_{0\leq t \leq T}$ is a cumulative consumption process.\\
Now we prove that $\pi(\xi)=\pi_0(\xi)$. Let $(x,H) \in \Sp(\xi)$. Then as $V_{T-1}^{x,H}+H_{T} \Delta S_{T} \geq \xi \ \Pc_{T-1}\mbox{-q.s.}$ it follows as in \eqref{eq. bn_lem4.10}
\begin{align*}
V_{T-1}^{x,H}(\omega)+H_{T}(\omega) \Delta S_{T}(\omega, \cdot) \geq \xi(\omega, \cdot) \; \Pc_{T-1}(\omega)\mbox{-q.s.}
\end{align*}
for all for $\omega\in \Omega^{T-1}$ in an $\Fc^{\Uc}_{T-1}$-measurable and $\Pc$-full measure set. From \eqref{defminr}, we conclude that $\pi_{T-1}(\xi)(\omega) \le V_{T-1}^{x,H}(\omega).$
By induction we see that $\pi_0(\xi) \le x$ and thus $\pi_0(\xi) \le \pi(\xi)$. Conversely, using \eqref{fantacool} and \eqref{beindidonc}
\begin{align*}
V_{T}^{\pi_0,\hat{H}} =\pi_0(\xi) + \sum_{t=1}^T \hat{H}_t \Delta S_t \geq \Ec^T(\xi)=\xi \quad  \Pc\mbox{-q.s.}
\end{align*}
and therefore $\pi_0(\xi)  \geq \pi(\xi)$. Thus $\Ec^{0}(\xi)=\pi_0(\xi)=\pi(\xi)$ by \eqref{beindidonc} and we obtain (recall \eqref{fantacool} and the definition of $\hat{C}$) that
\begin{align*}
V_t^{\pi(\xi),\hat{H},\hat{C}}=\Ec^t(\xi)=\Xt \quad  \Pc\mbox{-q.s.}
\end{align*}
Since $V_T^{\pi(\xi),\hat{H},\hat{C}}=\Ec^T (\xi)=\xi \quad \Pc\mbox{-q.s.}$,
$(\pi(\xi),\hat{H},\hat{C})$ is a superhedging strategy and it is also minimal. Indeed let $(x,H,C)\in \Sp(\xi)$ then $V_T^{x,H,C} \ge \xi$ $\Pc$-q.s. From \eqref{eeu2},
$V_t^{x,H,C} \ge \Xt=V_t^{\pi(\xi),\hat{H},\hat{C}}$ $\Pc$-q.s.
This concludes the proof.
\end{proof}
\newpage
\section{Proofs of \cref{thm. exis_new} and \cref{thm. unique_new} }
\label{sec:app_utility}

\subsection{Proof of \cref{thm. exis_new}: The one-period case}\label{Sec. one-per}
We now prove \cref{thm. exis_new} in the case $T=1$, where we  follow  arguments given in \cite{Nutz}. Let $\xi:\Omega^T \to \R$ be Borel. In preparation for the multi-period case we define the set
\begin{align*}
\mathcal{A}_{0,x}= \{(H,c) \in \R^d\times \R_+ \ |\  x-c+H\Delta S_1 \ge \pi_1(\xi)\  \Pc\text{-q.s.} \}.
\end{align*}
Recall definition $\pi_t(\xi)$ given in \eqref{defminr} for $t=0,1$ and note that if $(H,c)\in \mathcal{A}_{0,x}$ then also $(H,0)\in \mathcal{A}_{0,x}$. We thus often write $H\in \mathcal{A}_{0,x}$ instead of $(H,c)\in \mathcal{A}_{0,x}$. Let $U(1,\cdot, \cdot):\Omega \times [0,\infty) \to \R$ be bounded from above and $\mathcal{F}^{\mathcal{U}}_1$-measurable. Besides let us assume that 
%there exists a constant $a>0$ such that $\omega \mapsto U(1,\omega, a/2)$ is bounded from below and 
$x \mapsto U(1,\omega,x)$ is non-decreasing, concave and continuous for each $\omega \in \Omega$. Furthermore let the deterministic function $U(0, \cdot):[0,\infty) \to \R$ be non-decreasing and continuous. As usual we set $U(t,\omega, x)=-\infty$ for $x<0$ and $t=0,1$. Let us now state the main theorem for $T=1$:
\begin{proposition} \label{Thm Nutz1}
Let NA$(\Pc)$ hold and $x \ge \pi_0(\xi)$. 
%Assume that
%\begin{align} \label{eq integrability}
%\E_{P}[U^+(1,x+H\Delta S_1-\pi_1(\xi))] < \infty \ \text{ for all }H \in \mathcal{A}_{0,x} \ \text{and } P \in \Pc^u.
%\end{align}
%Furthermore we assume that 
%\begin{align}\label{eq. integrability2}
%\E_{P}[U^+(1,a)] <\infty \quad \text{for all }P \in \Pc^u.
%\end{align}
Then
\begin{align*}
u(x) := \sup_{(H,c) \in \mathcal{A}_{0,x}} \left(\inf_{P \in \Pc^u} \E_{P}[U(1,x-c+H\Delta S_1-\pi_1(\xi))]+U(0,c)\right)< \infty
\end{align*}
and there exists $(\hat{H},\hat{c}) \in \mathcal{A}_{0,x}$ such that $\inf_{P \in \Pc^u} \E_{P}[U(1,x-\hat{c}+\hat{H}\Delta S_1-\pi_1(\xi))]+U(0,\hat{c})=u(x).$
\end{proposition}

We prove the result via a lemma. Here we denote 
\begin{align*}
L=\text{span }(\{\text{supp}(P \circ (\Delta S_1)^{-1} \ | \ P \in \mathcal{P}\})\subseteq \R^d
\end{align*} and the orthogonal complement
\begin{align*}
L^{\perp}=\{H \in \R^d\ |\ H V=0 \text{ for all } V \in L\}.
\end{align*}

\begin{lemma}
Assume $x \ge \pi_0(\xi)$. Under NA$(\Pc)$ the set $K_{x}=\mathcal{A}_{0,x} \cap (L \times \R_+) \subseteq \R^{d+1}$ is non-empty, convex and compact.
\end{lemma}
\begin{proof}
Clearly $K_{x}$ is convex and closed. It remains to show that $K_{x}$ is bounded: As by definition of $\pi_0(\xi)$ clearly $c \in [0,x-\pi_0(\xi)]$ for all $c \in \mathcal{A}_{0,x}$ we only need to show that $H \in K_x$ is bounded. Note that after a translation by $(H_0,0)\in K_{x}$ we have $0 \in \tilde{K}_{x}:= K_{x}-(H_0,0)$. Now we assume towards a contradiction that there exist $H_n \in \tilde{K}_{x}$ such that $|H_n| \to \infty$. We define $\delta=|H_0|+1$. We can extract a subsequence $\delta H_n/|H_n|$ that converges to a limit $H \in \R^d$, so $|H|=\delta$. As $\tilde{K}_{x}$ is convex and contains the origin we have for $n$ large enough $\delta H_n/|H_n|\in \tilde{K}_{x}$.
It follows $H \in \tilde{K}_{x}$, since $\tilde{K}_{x}$ is closed. Furthermore
\begin{align*}
H\Delta S_1 \ge \liminf_{n \to \infty} \frac{\pi_1(\xi)-x-H_0\Delta S_1}{|H_n|/\delta}=0 \hspace{0.5cm} \ \Pc\text{-q.s.}
\end{align*}
By NA$(\Pc)$ this implies $H \Delta S_1=0$ $\Pc$-q.s. and thus $H \in L^{\perp}$ by use of \cite[Lemma 2.6]{Nutz}. As $H\in \tilde{K}_{x}$ this implies $H_0+H \in K_{x} \subseteq L$, which means $|H|^2=-H_0H$. This contradicts $|H|=\delta$ by Cauchy-Schwarz inequality. 
\end{proof}

\begin{proof}[of \cref{Thm Nutz1}]
%\Cref{lem integrability} and 
Fatou's lemma 
%(see \eqref{eq integrability}, \eqref{eq. integrability2}) 
implies that for all $P \in \Pc^u$ the function $(H,c) \mapsto \E_{P}[U(1,x-c+H\Delta S_1-\pi_1(\xi))]+U(0,c)$ is upper semicontinuous on $\mathcal{A}_{0,x}$. It follows that $(H,c) \mapsto \inf_{P \in \Pc^u} \E_{P}[U(1,x-c+H\Delta S_1-\pi_1(\xi))]+U(0,c)$ is upper semicontinuous and thus attains its supremum on the compact set $K_{x}$. Finally again using \cite[Lemma 2.6]{Nutz} and recalling that $\Pc^u \subseteq \Pc$
\begin{align*}
&\sup_{(H,c) \in {\mathcal{A}}_{0,x}} \left(\inf_{P \in \Pc^u} \E_{P}[U(1,x-c+H\Delta S_1-\pi_1(\xi))]+U(0,c)\right)\\
=\ &\sup_{(H,c) \in K_{x}} \left(\inf_{P \in \Pc^u} \E_{P}[ U(1,x-c+H\Delta S_1-\pi_1(\xi))]+U(0,c)\right).
\end{align*}
\end{proof}

\begin{corollary} \label{cor. minim}
Under the conditions of \cref{Thm Nutz1} we have
\begin{align*}
&\sup_{(H,c) \in \mathcal{A}_{0,x}} \left(\inf_{P \in \Pc^u} \E_{P}[U(1,x-c+H\Delta S_1-\pi_1(\xi))]+U(0,c)\right)\\
=\ &\inf_{P \in \Pc^u}\left( \sup_{(H,c) \in \mathcal{A}_{0,x}} \left(\E_{P}[U(1,x-c+H\Delta S_1-\pi_1(\xi))]+U(0,c)\right)\right).
\end{align*}
\end{corollary}
\begin{proof}
Note that $K_x$ is compact, convex and $\Pc^u$ is convex. Define
\begin{align*}
f: K_x \times \mathfrak{P}(\Omega) \to \R \hspace{0.5cm} (H,c,P) \mapsto  \E_{P}[U(1,x-c+H\Delta S_1-\pi_1(\xi))]+U(0,c)
\end{align*}
and note that $(H,c) \mapsto f(H,c,P)$ is upper semicontinuous and concave. Furthermore $P \mapsto f(H,c, P)$ is convex on $\Pc^u$. The claim follows from Corollary 2 in \cite{terk}.
\end{proof}

\begin{remark}
The boundedness from above of $U(1,\cdot, \cdot)$ can be replaced by a weaker condition: Indeed it is sufficient to assume there exists a constant $a>0$ such that $\omega \mapsto U(1,\omega, a/2)$ is bounded from below and 
\begin{align*}
\E_{P}[U^+(1,x+H\Delta S_1-\pi_1(\xi))] < \infty \ \text{ for all }H \in \mathcal{A}_{0,x} \ \text{and } P \in \Pc^u
\end{align*}
as well as
\begin{align*}
\E_{P}[U^+(1,a)] <\infty \quad \text{for all }P \in \Pc^u.
\end{align*}
The proof of \cref{Thm Nutz1} then follows along the lines of  \cite[Lemma 1]{RS06} and \cite[Lemma 2.8]{Nutz} after a translation by $H_0\in \text{ri}(K_{x})$.
\end{remark}

\subsection{Proof of \cref{thm. exis_new}: The multi-period case}\label{Sec Mulit}

For the rest of this section we assume NA$(\Pc)$ and that $\xi$ is Borel measurable. 
%We write $\omega \otimes \omega'=(\omega, \omega') \in \Omega^{t+1}$ for $\omega \in \Omega^t$ and $\omega' \in \Omega$.
Furthermore we often abbreviate $\Xt$ by $\pi_t$. To simplify notation we assume $U(0,\cdot,0)=0$. We give the following definition:

\begin{definition} \label{Def U_t}
We define $U_T(\omega,x)=U(T,\omega, x)$ and for $0\le t \le T-1$
\begin{align*}
U_{t}(\omega, x) &:= \sup_{(H,c) \in \mathcal{A}_{t,x}(\omega)}\bigg( \inf_{P \in \Pc_{t}^u(\omega)} \E_{P}[U_{t+1}((\omega,\cdot), x+H\Delta S_{t+1}(\omega, \cdot)-c-\mathds{1}_{\{t=T-1\}}\xi(\omega, \cdot))]\\
&\ +U(t,\omega,c)\bigg),\quad  x\ge \pi_{t}(\omega) 
\end{align*}
and $U_t(\omega,x)=-\infty$ otherwise, where for $x \in \R$ we set
\begin{align*}
\mathcal{A}_{0,x}(\omega) &:= \{ (H,c) \in \R^d\times \{0\} \ | \ x+H\Delta S_{1}(\omega,\cdot) \ge \pi_{1}(\omega, \cdot) \ \Pc_0(\omega)\text{-q.s.} \} \\
\mathcal{A}_{t,x}(\omega) &:= \{ (H,c) \in \R^d\times\R_+ \ | \ x+H\Delta S_{t+1}(\omega,\cdot)-c \ge \pi_{t+1}(\omega, \cdot) \ \Pc_t(\omega)\text{-q.s.} \}, \quad t \ge 1.
\end{align*}
\end{definition}

We recall from \cref{lem::usa} that $\pi_t(\xi)$ is upper semianalytic. This means in particular that $$\{(\omega,x) \ | \ x < \pi_t(\xi)(\omega) \}=\bigcup_{q\in \QQ} \pi_t^{-1}((q,\infty))\times (-\infty,q)$$ is analytic. Next we show by backwards induction, that if \cref{Ass 1a} is satisfied, then $U_t$ has $\Pc^u$-q.s. the following properties:
\begin{condition} \label{Cond 2}
Let $0\le t \le T-1$. The function $U_t: \Omega^t \times \R \to [-\infty, \infty)$ is lower semianalytic and bounded from above. Furthermore the following properties hold:
\begin{enumerate}
\item $\omega \mapsto U_t(\omega, x(\omega))$ is bounded from below for $x(\omega): = \pi_{t}(\omega)+\epsilon$ and each $\epsilon >0$.
\item $x \mapsto U_t(\omega,x)$ is non-decreasing, concave and continuous on $[\pi_t(\omega),\infty)$ for each $\omega \in \Omega^t$.
\end{enumerate}
\end{condition}

\begin{lemma}\label{lem usa}
Let NA$(\Pc)$ and \cref{Qanalytic}, \cref{Ass 1b} and \cref{Ass 1a}  hold for $U (t, \cdot, \cdot)$, $0 \le t\le T$. Then there exist functions $\tilde{U}_t: \Omega^t \times (-\infty, \infty) \to [-\infty, \infty)$, which satisfy \cref{Cond 2}, such that $\tilde{U}_t= U_t$ $\Pc^u$-q.s.
\end{lemma}

\begin{proof}
We prove the claim by induction. Recall that $U_T$ satisfies \cref{Ass 1a}. We now show the induction step from $t+1$ to $t$ and therefore first fix $\omega \in \Omega^t$. For simplicity of presentation we assume $t \le T-2$.\\
We first state some results regarding lower semianalyticity, which lead to the definition of $\tilde{U}_t$: Using \cite[Lemma 7.30, p.177, Prop. 7.47, p.179, Prop. 7.48, p.180]{bs}, \cref{Ass 1a} and the analytic graph of $\Pc^u_t$ we see that
$\phi: \Omega^t \times (-\infty, \infty) \times \R^d \times \R\to \overline{\R}$
\begin{align*}
\phi(\omega,x,H,c)= \inf_{P \in \Pc^u_t(\omega)} \E_{P}[U_{t+1}((\omega, \cdot), x+H\Delta S_{t+1}(\omega, \cdot)-c)]+U(t,\omega,c)
\end{align*}
is lower semianalytic as $\Delta S_{t+1}(\omega, \cdot)$ is a Borel measurable functions (and also $\xi(\omega, \cdot)$ for $t=T-1$). Now we define the function $\tilde{\phi}: \Omega^t\times \R \times \R^d \times \R\to \overline{\R}$ 
\begin{align*}
\widetilde{\phi}(\omega,x,H,c)=
\left\{\begin{array}{ll}
					-\infty &\text{if }(H,c) \notin \mathcal{A}_{t,x} \text{ or }x<\pi_t(\xi)(\omega)\\
					\phi(\omega,x,H,c) &\text{otherwise.} 
			\end{array}
				\right. 
\end{align*}
We show that $\tilde{\phi}$ is lower semianalytic. Fix $a \in \R$. Then 
\begin{align*}
\left\{\widetilde{\phi}<a\right\}&= \{(\omega,x,H,c) \ | \ \phi(\omega,x,H,c) < a, \ (H,c) \in \mathcal{A}_{t,x}(\omega), x\ge \pi_t(\xi)(\omega)\} \\
&\cup \left\{(\omega,x,H,c) \ | \ (H,c) \notin \mathcal{A}_{t,x}(\omega) \text{ or }x< \pi_t(\xi)(\omega)\right\} \\
&= \{\phi < a \} \cup \left\{(\omega,x,H,c) \ | \ (H,c) \notin \mathcal{A}_{t,x}(\omega)\right\} \\
&\cup \left\{(\omega,x,H,c) \ | \ x < \pi_t(\xi)(\omega)\right\}. 
\end{align*}
By the same arguments as for the lower seminanalyticity of $\phi$ we see that
\begin{align*}
&\left\{(\omega,x,H,c) \ | \ (H,c) \notin \mathcal{A}_{t,x}(\omega)\right\}
\\
=\ &\Bigg\{(\omega,x,H,c)\ \bigg| \ \sup_{P \in \Pc^u_t(\omega)}  \E_{P}[x+ H\Delta S_{t+1}(\omega, \cdot)-c-\pi_{t+1}(\omega, \cdot)]^- > 0\Bigg\}
\end{align*}
is analytic and the sets
\begin{align*}
\{ \phi < a \} \qquad \text{and} \qquad
\{ (\omega,x,H,c) \in \Omega^t \times\R \times \R^d \times \R\ | \ x < \pi_t(\xi)(\omega) \} 
\end{align*}
are analytic, so $\widetilde{\phi}$ is lower semianalytic. 
Similarly to \cite[Proposition 3.27]{BC16} we define
\begin{align*}
\tilde{U}_t(\omega,x)=\lim_{n \to \infty}\sup_{(H,c) \in \QQ^d\times \QQ_+} \widetilde{\phi}\left(\omega,x+\frac{1}{n},H,c\right).
\end{align*}
As the limits and countable supremum of lower semianalytic functions is lower semianalytic, we conclude that $\tilde{U}_t$ is lower semianalytic. \\
From the definition it is clear that $\tilde{U}_t(\omega, \cdot)$ is non-decreasing and bounded from above. Next we argue that $\tilde{U}_t(\omega,\cdot)$ is concave. As the infimum of concave functions is concave, it is enough to argue that $x \mapsto \sup_{(H,c) \in \QQ^d\times \QQ_+} \widetilde{\phi}\left(\omega,x,H,c\right)$ is concave. This follows very similarly to  \cite{RS06}[proof of Prop. 2, p.5]: Indeed, it is enough to show midpoint-concavity of $\sup_{(H,c) \in \QQ^d\times \QQ_+} \widetilde{\phi}\left(\omega,\cdot,H,c\right)$, which is immediate by use of triangle inequality. Concavity implies that $ \tilde{U}_t(\omega, \cdot)$ is continuous on $(\pi_t(\omega), \infty)$. By the definition of $\tilde{U}_t$ concavity and continuity extend to $[\pi_t(\omega), \infty)$. \\
%as well as the function
%\begin{align*}
%\tilde{U}_t(\omega,y)=\inf_{n \in \N} \sup_{x \le \pi_t(\omega)+1/n} \widehat{\phi}(\omega,x,y).
%\end{align*}
%In order to show that $\tilde{U}_t$ is lower semianalytic we fix $n \in \N$ and define the correspondence
%\begin{align*}
%\Lambda: \Omega^t \twoheadrightarrow \R,\hspace{0.5cm} \Lambda(\omega)=(-\infty, \pi_t(\omega)+1/n],
%\end{align*}
%which is analytically measurable and closed-valued, so we can approximate the set
%\begin{align*}
%&\left\{(\omega,y)\in \Omega^t \times \R \ \bigg| \ \sup_{x \le \pi_t(\omega)+1/n} \widehat{\phi}(\omega,x,y)\le c \right\}
%\end{align*}
%for $c \in \R$ using the Castaing representation of $\Lambda$ and the fact that graph($\Lambda$) is analytic in order to see that
%$\tilde{U}_t(\omega,y)$ is again lower semianalytic.\\
By definition we clearly have
\begin{align*}
\sup_{(H,c) \in \mathbb{Q}^d\times \QQ_+} \widetilde{\phi}\left(\omega,x,H,c\right)\le \sup_{(H,c) \in \Ac_{t,x}(\omega)} \phi(\omega,x,H,c).
\end{align*} 
We now show equality of $U_t(\omega,x)$ and $\tilde{U}_t(\omega,x)$ for $\Pc^u$-q.e. $\omega \in \Omega^t$. 
%\begin{align}\label{eq. own}
%\widehat{\phi}(\omega,x,y)&=\lim_{n \to \infty}\sup_{H \in \QQ^d} \widetilde{\phi}\left(\omega,x+\frac{1}{n},y+\frac{1}{n},H\right)\nonumber\\
%&=\left\{\begin{array}{ll}
%					\sup_{H \in \mathcal{H}_{t,x}(\omega)} \phi(\omega,y,H). &\text{ if }y \ge \pi_t(\omega), \\
%					-\infty &\text{ otherwise. }
%			\end{array}
%				\right. \
%\end{align} 
Let us therefore fix $x > \pi_t(\omega)$ and $\omega \in \Omega_{\text{NA}}^t$. Using \cite[Theorem 3.4]{BN} and $\Pc^u_t(\omega) \subseteq \Pc_t(\omega)$ there exists $\tilde{H} \in \R^d$ such that
\begin{align*}
\pi_t(\omega)+ \tilde{H}\Delta S_{t+1}(\omega,\omega') \ge \pi_{t+1}(\omega, \omega')\quad \text{ for }\Pc_t^u(\omega)\text{-q.e. }\omega' \in \Omega.
\end{align*}
Take $c<x-\pi_t(\omega)$ and $H \in [0, \infty)^d$ such that 
\begin{align*}
H^1+ \dots +H^d \le \frac{x-\pi_t(\omega)-c}{\max_{1 \le i \le d}S^i_t(\omega)}.
\end{align*}
It follows for $\Pc^u_t(\omega)$-q.e. $\omega'\in \Omega$ that
%\marginpar{Here $S_t \ge 0$ is needed.}
\begin{align*}
x+(H+\tilde{H})\Delta S_{t+1}(\omega, \omega')-c &=x -\pi_t(\omega)+ H \Delta S_{t+1}(\omega, \omega')+ \pi_t(\omega)
+\tilde{H}\Delta S_{t+1}(\omega, \omega')-c\\
&\ge x-\pi_t(\omega) -H S_t(\omega)
+\pi_{t+1}(\omega, \omega')-c \\&\ge \pi_{t+1}(\omega, \omega').
\end{align*}
Thus the affine hull of $\mathcal{A}_{t,x}(\omega)$ is $\R^{d+1}$ and consequently $\text{Ri}(\mathcal{A}_{t,x}(\omega))$ is an open set in $\R^{d+1}$. This implies
\begin{align*}
\sup_{(H,c) \in \mathbb{Q}^d\times \QQ_+} \widetilde{\phi}\left(\omega,x,H,c\right)= \sup_{(H,c) \in \Ac_{t,x}(\omega)} \phi(\omega,x,H,c).
\end{align*} 
for $x > \pi_t(\omega)$. Equality in $x=\pi_t(\omega)$ follows by right-continuity of $U_t$ and $\tilde{U}_t$. Indeed, right-continuity of $U_t(x,\omega)$ in $x=\pi_{t}(\omega)$ follows by compactness of $\mathcal{A}_{t,\pi_t(\omega)+1}(\omega)\cap \text{span}(\text{supp}(\{P \circ (\Delta S_{t+1}(\omega, \cdot)^{-1} \ | \ P \in \mathcal{P}_t(\omega)\}))$ and Fatou's Lemma.\\
%For any $n\in \N$ we can find maximizers $\hat{H}_n\in K_{t,\pi_t(\omega)+1/n} \subseteq K_{t,\pi_t(\omega)+1}$ by the One-Period result. Since $K_{t,\pi_t(\omega)+1}$ is compact, we can assume without loss of generality $\hat{H}_n \to \hat{H} \in K_{t, \pi_t(\omega)}$. So by Fatou's lemma
%\begin{align*}
%\lim_{n \to \infty} \sup_{H \in \mathcal{H}_{t,\pi_t(\omega)+1/n}}  \phi(\omega,y,H) &= \lim_{n \to \infty} \inf_{P \in \Pc^u_t(\omega)} \E_{P}(U_{t+1}((\omega, \cdot), y+\hat{H}_n\Delta S_{t+1}(\omega, \cdot))) \\
%&\le \inf_{P \in \Pc^u_t(\omega)} \E_{P}(U_{t+1}((\omega, \cdot), y+\hat{H}\Delta S_{t+1}(\omega, \cdot)))\\
%&\le  \sup_{H \in \mathcal{H}_{t,\pi_t(\omega)}} \phi(\omega,y,H). 
%\end{align*}
%This shows right-continuity of $x \mapsto  \sup_{H \in \mathcal{H}_{t,x}}\phi(\omega,y,H)$ in $x=\pi_t(\omega)$. 
%It is elementary to show concavity of $x \mapsto  \sup_{H \in \mathcal{H}_{t,x}}\phi(\omega,y,H)$. Consequently we see that \eqref{eq. own} holds. Lastly we note that by \eqref{eq. own}
%\begin{align*}
%U_t(\omega, y) &=\left\{\begin{array}{ll}\sup_{H \in \mathcal{H}_{t,\pi_t(\omega)}(\omega)} \inf_{P \in \Pc^u_t(\omega)} \E_{P}(U_{t+1}((\omega, \cdot), y+H\Delta S_{t+1}(\omega, \cdot)))  &\text{if }y \ge \pi_t(\omega), \\
%-\infty & \text{otherwise} \end{array}\right.\\
%&=\left\{\begin{array}{ll}\sup_{H \in \mathcal{H}_{t,\pi_t(\omega)}(\omega)} \phi(\omega,y,H) &\text{if }y \ge \pi_t(\omega), \\
%-\infty & \text{otherwise} \end{array}\right.
%=\inf_{n \in \N} \sup_{x \le \pi_t(\omega)+1/n} \widehat{\phi}(\omega,x,y)=\tilde{U}_t(\omega,y)
%\end{align*}
%$\Pc^u$-q.s.
Lastly we show boundedness of $\tilde{U}_t$ from below: Let $x(\omega)=\pi_t(\omega)+\epsilon$ for some $\epsilon>0$. By the above arguments there exists $\hat{H}\in \QQ^d$ such that $\pi_t(\omega)+\epsilon/3+\hat{H}\Delta S_{t+1}(\omega, \omega')\ge \pi_{t+1}(\omega, \omega')$ $\mathcal{P}^u_t(\omega)$-a.s. Thus
\begin{align*}
U_{t}(\omega,x(\omega)) &\ge \inf_{P \in \mathcal{P}_t^u(\omega)} \E_{P}[U_{t+1}((\omega, \cdot), x(\omega)+\hat{H} \Delta S_{t+1}(\omega, \cdot)-\epsilon/3)] +U(t,\omega,\epsilon/3)\\
&\ge \inf_{P \in \mathcal{P}_t^u(\omega)} \E_{P}[U_{t+1}((\omega, \cdot), \pi_{t+1}(\omega, \cdot)+\epsilon/3)]+U(t, \omega,\epsilon/3)
\end{align*}
is bounded from below by the induction hypothesis and \cref{Ass 1a}. This shows the claim.\\
\end{proof}

%The next two lemmas are auxiliary results, which we state here without proof.
%\marginpar{JW: To do: Quote in proof of Lemma \cref{lemusc} - under construction}
%
%\begin{lemma}[\cite{BN} Lemma 4.11, page 847]\label{lem nutz1}
%Let $X,Y$ be Polish spaces and $D \in\textbf{A}(\mathcal{U}(X)\otimes \mathcal{B}(Y))$, where $\mathcal{U}(X)$ denotes the universal completion of $\Bc(X)$. Then $\text{proj}_{X}(D) \in \mathcal{U}(X)$ and there exists an $\mathcal{U}(X)$-measurable mapping $\phi: \text{proj}_{X} (D) \to Y$ such that $\text{graph}(\phi) \subseteq D$.
%\end{lemma}
%
%The following result is a correction to \cite{BN}, Lemma 4.12, page 848, which was identified by \cite{BC16}, Lemma 7.6 (see also their counterexample 7.5).
%\begin{lemma}(\cite{BC16}, Lemma 7.6, page 41)\label{lem cara} Let $(X, \mathcal{X})$ be a measurable space and $f: \R^d \times X \to \overline{\R}$ be a function such that $\omega \mapsto f(y, \omega)$ is $\mathcal{X}$-measurable for all $y \in \R^d$ and $y \mapsto f(y, \omega)$ is upper semicontinuous and concave for all $\omega \in X$. Then $f$ is $\mathcal{B}(\R^d) \otimes \mathcal{X}$-measurable.
%\end{lemma}

\begin{lemma}\label{lem meas_selec}
Let NA$(\Pc)$ and \cref{Qanalytic}, \cref{Ass 1b} and \cref{Ass 1a}  hold for $U (t, \cdot, \cdot)$, $0 \le t\le T$. Let $t \in \{0, \dots, T-1\}$ and $(H,C) \in \mathcal{A}_{\pi_0}$. There exist universally measurable mappings $\hat{H}_{t+1},\hat{c}_t$ such that $\hat{c}_t$ is non-negative,
\begin{align*}
V^{\pi_0,H,C}_{t-1}(\omega)+H_t(\omega)\Delta S_t(\omega)+ \hat{H}_{t+1}(\omega)\Delta S_{t+1}(\omega,\cdot)-\hat{c}_t(\omega) \ge \pi_{t+1}(\omega,\cdot)\ \Pc_t(\omega)\text{-q.s.}
\end{align*}
and
\begin{align*}
\inf_{P \in \Pc_t^u(\omega)} &\E_P\left[U_{t+1}\left((\omega, \cdot), V^{\pi_0,H,C}_{t-1}(\omega)+H_t(\omega)\Delta S_t(\omega) + \hat{H}_{t+1}(\omega) \Delta S_{t+1}(\omega, \cdot)-\hat{c}_t(\omega)-\mathds{1}_{\{t=T-1\}}\xi(\omega, \cdot)\right)\right]\\&+U(t,\omega,\hat{c}_t(\omega))
= U_t\left(\omega, V^{\pi_0,H,C}_{t-1}(\omega)+H_t(\omega)\Delta S_t(\omega)\right)
\end{align*}
for $\Pc^u$-a.e. $\omega \in \Omega^t$.
\end{lemma}

\begin{proof}
We show that $\tilde{U}_t$ is $\mathcal{F}^{\mathcal{U}}_t \otimes \mathcal{B}(\R)$-measurable: Indeed, we know that $\omega \mapsto \tilde{U}_t(\omega,x)$ is lower seminanalytic and in particular universally measurable. Also $x \mapsto \tilde{U}_t(\omega,x)$ is continuous on $[\pi_t(\omega), \infty)$, bounded from above and $\tilde{U}_t(\omega,x)=-\infty$ for $x< \pi_t(\omega)$. Thus is it concave and upper semicontinuous on $\R$ and the claim follows from \cite[Lemma A.35, p. 1889]{BC16}.
% Furthermore 
%\begin{align*}
%A := \{ (\omega,x) \in \Omega^t \times \R \ | \ x \ge \pi_t(\omega) \}
%\end{align*}
%is $\mathcal{F}^{\mathcal{U}}_t \otimes \mathcal{B}(\R)$-measurable by Lemma \cref{lem cara}. Let $F \subseteq \R^d$ be closed, then 
%\begin{align*}
%\tilde{U}_t^{-1}(F)= \bigcap_{n=1}^{\infty}\bigcup_{q \in \QQ} \left(\{ \omega \in \Omega^t \ | \ \text{dist}(\tilde{U}_t(\omega,q),F) <1/n \} \times B_{\frac{1}{n}}(q) \right)\cap A, 
%\end{align*}
%so the claim follows. 
Next we show that the function
\begin{align*}
\phi(\omega,x,H,c)= \inf_{P \in \Pc^u_t(\omega)} \E_{P}[\tilde{U}_{t+1}((\omega, \cdot), x+H\Delta S_{t+1}(\omega, \cdot)-c)]+U(t,\omega,c)
\end{align*}
is $\mathcal{F}^{\mathcal{U}}_t \otimes \mathcal{B}(\R) \otimes \mathcal{B}(\R^d)\otimes \mathcal{B}(\R)$-measurable: As we have argued in \cref{lem usa} $\omega \mapsto \phi(\omega,x,H,c)$ is lower semianalytic and in particular universally measurable. On the other hand, $x \mapsto \tilde{U}_{t+1}(\omega, x)$ is upper semicontinuous and concave for any $\omega \in \Omega^t$. Since $\tilde{U}_{t+1}$ is bounded from above, an application of Fatou's lemma yields that $(x,H,c) \mapsto \phi(\omega,x,H,c)$ is upper semicontinuous and concave for each $\omega \in \Omega^t$. Again by \cite[Lemma A.35, page 1889]{BC16} it follows that $\phi$ is $\mathcal{F}^{\mathcal{U}}_t \otimes \mathcal{B}(\R) \otimes \mathcal{B}(\R^d)\otimes \mathcal{B}(\R)$-measurable.
Now we define the correspondence
\begin{align*}
\Phi(\omega) :&= \{ (H',c') \in \R^d\times \R_+ \ | \ \phi(\omega, V^{\pi_0,H,C}_{t-1}(\omega)+H_t(\omega) \Delta S_{t}(\omega)-\mathds{1}_{\{t=T-1\}}\xi(\omega, \cdot), H',c')\\
&\qquad = \tilde{U}_t(\omega,V^{\pi_0,H,C}_{t-1}(\omega)+H_{t}(\omega) \Delta S_t (\omega)) \} , \quad \omega \in \Omega^t.
\end{align*}
Then its graph is in $\Fut \otimes \mathcal{B}(\R^d)\otimes\mathcal{B}(\R)$. Next we define the function 
$$\Upsilon :\omega \mapsto \mathcal{A}_{t,V_{t-1}^{\pi_0,H,C}(\omega)+H_t(\omega) \Delta S_t(\omega)}(\omega).$$ 
By a slight variation of the arguments given in \cite{BN}[proof of Lemma 4.10, pp.846-848] the graph of $\Upsilon$ is $\Fc^{\mathcal{U}}_t \otimes \mathcal{B}(\R^d)\otimes \mathcal{B}(\R)$-measurable and thus $
\text{graph}(\Upsilon) \cap ((\Omega_{NA}^t \cap \Omega_{\xi}^t)\times \R^d \times \R)\in\mathcal{F}_{t}^{\mathcal{U}} \otimes \Bc(\R^d) \otimes \Bc(\R)$. Then also the graph of 
\begin{align*}
\tilde{\Phi}(\omega)=\begin{cases} \mathcal{A}_{t,V_{t-1}^{\pi_0,H,C}(\omega)+H_t(\omega) \Delta S_t(\omega)}(\omega) \cap \Phi(\omega) & \omega \in \Omega_{NA}^t \cap \Omega_{\xi}^t\\
\emptyset & \text{otherwise}
\end{cases}
\end{align*}
is in $\mathcal{F}_{t}^{\mathcal{U}} \otimes \Bc(\R^d) \otimes \Bc(\R)$ and $\tilde{\Phi}$ admits an $\mathcal{F}_t^{\mathcal{U}}$-measurable selector $(\hat{H}_{t+1},\hat{c}_t)$ on the universally measurable set $\{ \tilde{\Phi} \neq \emptyset \}\in \mathcal{F}_t^{\Uc}$ by the Neumann-Aumann theorem (\cite[Cor.1, p.120]{bv}). We extend $(\hat{H}_{t+1}, \hat{c}_t)$ by setting $\hat{H}_{t+1}= \hat{c}_t=0$ on $\{\tilde{\Phi} \neq \emptyset \}$. Moreover the one-period case given in \cref{Thm Nutz1} applied with $x=V_{t-1}^{\pi_0,H,C}(\omega)+H_t(\omega) \Delta S_t(\omega)$ , \cref{olala}, \cref{rem. laurence} as well as existence of superhedging strategies as stated in \cite[Theorem 3.4]{BN} show that $\tilde{\Phi}(\omega) \neq \emptyset$ for $\Pc^u$-q.e. $\omega \in \Omega^t$. This shows the claim as $U_t= \tilde{U}_t$ $\Pc^u$-q.s.
\end{proof}
%The remaining steps of the proof is the same as in %\cite{Nutz}.

%\begin{lemma}\label{Lem max2}
%Let $t \in \{0, \dots, T-1\}$ and $(H,C) \in \mathcal{A}_{\pi_0}$. Define
%\begin{align*}
%I_t (\omega)&= \inf_{P \in \Pc^u_t(\omega)} \E_{P}(U_{t+1}((\omega, \cdot), V^{\pi_0,H,C}_{t}(\omega)+H_{t+1}\Delta S_{t+1}-\mathds{1}_{\{t=T-1\}}\xi(\omega, \cdot)))\\
%&+\sum_{s=1}^{t} U(s,\omega,\Delta C_s(\omega)), \ \ \omega \in \Omega^t.
%\end{align*}
%Given $\epsilon >0$, there exists a universally measurable kernel $P_t^{\epsilon}: \Omega^t \to \mathfrak{P}(\Omega)$ such that $P_t^{\epsilon}(\omega) \in \Pc^u_t(\omega)$ for all $\omega \in \Omega^t$ and
%\begin{align*}
%&\E_{P_t^{\epsilon}}(U_{t+1}((\omega, \cdot), V^{\pi_0,H,C}_{t}(\omega)+H_{t+1}\Delta S_{t+1}-\mathds{1}_{\{t=T-1\}}\xi(\omega, \cdot))) +\sum_{s=1}^{t} U(s,\omega,\Delta C_s(\omega))\\
%\le \ \ & 
%\left\{\begin{array}{ll}
%					I_t(\omega)+\epsilon &\text{ if }I_t(\omega) > -\infty, \\
%					(-\epsilon)^{-1} &\text{ if } I_t(\omega)= -\infty.
%			\end{array}
%				\right. 
%\end{align*}
%\end{lemma}
%\begin{proof}
%This follows as in \cite[Lemma 3.8, p.15]{Nutz}, recalling that the function $(\omega, P, x, H, c) \mapsto \E_{P}(U_{t+1}((\omega, \cdot), x+H\Delta S_{t+1}(\omega, \cdot)-c)+\sum_{s=1}^{t} U(s,\omega,\Delta C_s(\omega))$ is lower semianalytic. Then one concludes by an application of the Jankov-von-Neumann theorem as stated in \cite[Prop. 7.50, p. 184]{bs} and \cite[Prop. 7.44, p.172]{bs}.
%\end{proof}
\begin{proof}[of \cref{thm. exis_new}]
Let $(\hat{H}_1,0)$ be an optimal strategy for $$
\inf_{P \in \Pc^u_0} \E_P(U_{1}[ \pi_0+ H_{1} \Delta S_{1})]$$ as in \cref{lem meas_selec}. Proceeding recursively, we use \cref{lem meas_selec} to define the strategy $\omega \mapsto (\hat{H}_{t+1},  \hat{c}_t)(\omega)$ for 
\begin{align*}
\inf_{P \in \Pc^u_t(\omega)} &\E_{P}[U_{t+1}((\omega, \cdot), V^{\pi_0,\hat{H}, \hat{C}}_{t-1}(\omega)+\hat{H}_{t}(\omega)\Delta S_t(\omega)+ H_{t+1}(\omega) \Delta S_{t+1}(\omega, \cdot)-c_t(\omega) -\mathds{1}_{\{t=T-1\}}\xi(\omega, \cdot))]\\
&+U(t, c_t(\omega))
\end{align*}
where $1\le t \le T-1$ and define $\hat{C}_t=\sum_{s=1}^t \hat{c}_s$ as well as $\Delta\hat{C}_T=V_{T-1}^{\pi_0,\hat{H}, \hat{C}}+\hat{H}_T \Delta S_T-\xi$. By construction we then have $(\hat{H}, \hat{C}) \in \mathcal{A}_{\pi_0}$. To establish that $(\hat{H}, \hat{C})$ is optimal we first show that
\begin{align}\label{eq own2}
\inf_{P \in \Pc^u} \E_{P}\left[\sum_{s=1}^T U(s,\Delta\hat{C}_s)\right] \ge U_0(\pi_0).
\end{align}
Let $0\le t \le T-1.$ By definition of $(\hat{H}, \hat{C})$ we have
\begin{align*}
&\inf_{P' \in \Pc^u_t(\omega)} \E_{P'}[U_{t+1}((\omega, \cdot),V^{\pi_0,\hat{H}, \hat{C}}_{t}(\omega)+\hat{H}_{t+1}(\omega)\Delta S_{t+1}(\omega, \cdot)-\mathds{1}_{\{t=T-1\}}\xi(\omega, \cdot))]\\
+\ &U(t,\omega, \Delta \hat{C}_t(\omega)) =U_t(\omega, V^{\pi_0,\hat{H}, \hat{C}}_{t-1}(\omega)+\hat{H}_t(\omega) \Delta S_t(\omega))
\end{align*}
for all $\omega \in \Omega^t$ outside a $\Pc^u$-polar set. Let $P \in \mathfrak{P}$, then $P= P_0 \otimes \cdots \otimes P_{T-1}$ for some selectors $P_t$ of $\Pc^u_t$, $t=0, \dots, T-1$ and we conclude via Fubini's theorem that
\begin{align*}
&\E_{P}\left[U_{t+1}\left(V^{\pi_0,\hat{H}, \hat{C}}_{t}+\hat{H}_{t+1}\Delta S_{t+1}-\mathds{1}_{\{t=T-1\}}\xi\right)+ \sum_{s=1}^t U(s, \Delta \hat{C}_s)\right]\\
=\ &\E_{(P_0 \otimes \cdots P_{t-1})(d\omega)}\bigg( \E_{P_t(\omega)}\left[U_{t+1}\left((\omega, \cdot), V^{\pi_0,\hat{H}, \hat{C}}_{t}(\omega)+\hat{H}_{t+1}(\omega)\Delta S_{t+1}(\omega, \cdot)-\mathds{1}_{\{t=T-1\}}\xi(\omega, \cdot)\right)\right]\\
+\ &\sum_{s=1}^t U\left(s,\omega, \Delta \hat{C}_s(\omega)\right)\bigg)\\
 \ge \ &\E_{P_0 \otimes \cdots \otimes P_{t-1}}\bigg[U_{t}\left(V^{\pi_0,\hat{H}, \hat{C}}_{t-1}+\hat{H}_t \Delta S_t\right)+ \sum_{s=1}^{t-1} U\left(s,\Delta \hat{C}_s\right) \bigg]\\
=\ &\E_{P}\bigg[U_t\left(V^{\pi_0,\hat{H}, \hat{C}}_{t-1}+\hat{H}_t \Delta S_t\right)+ \sum_{s=1}^{t-1} U\left(s, \Delta \hat{C}_s\right)\bigg].
\end{align*}
A repeated application of this inequality shows \eqref{eq own2}. To conclude that $(\hat{H}, \hat{C})$ is optimal, it remains to prove that
\begin{align*}
U_0(\pi_0) \ge \sup_{(H,C) \in \mathcal{A}_{\pi_0}} \inf_{P \in \Pc^u} \E_{P}\left[\sum_{s=1}^T U(s,\Delta C_s)\right]  =:v(\pi_0).
\end{align*}
To this end we fix an arbitrary $(H,C) \in \mathcal{A}_{\pi_0}$ and first show that
\begin{align}\label{eq max}
&\inf_{P \in \Pc^u} \E_{P}\left[U_t\left(V^{\pi_0,H,C}_{t-1}+H_t \Delta S_t \right) +\sum_{s=1}^{t-1}U(s, \Delta C_s)\right]\\\nonumber
\ge\ &\inf_{P \in \Pc^u} \E_{P}\left[U_{t+1}\left(V^{\pi_0,H,C}_{t}+H_{t+1} \Delta S_{t+1}-\mathds{1}_{\{t=T-1\}}\xi\right)+ \sum_{s=1}^t U(s, \Delta C_s)\right], \quad  t=1, \dots, T-1.
\end{align}
Let $\epsilon>0$. As in the proof of \cref{lem usa}
\begin{align*}
(\omega, P) \mapsto \E_{P}\left[U_{t+1}((\omega, \cdot), V^{\pi_0,H,C}_{t}(\omega)+H_{t+1}\Delta S_{t+1}-\mathds{1}_{\{t=T-1\}}\xi(\omega, \cdot))\right]+\sum_{s=1}^{t} U(s,\omega,\Delta C_s(\omega)),
\end{align*}
is lower semianalytic. Using \cite[Prop. 7.50, p. 184]{bs} and \cite[Prop. 7.44, p.172]{bs} for $\omega \in \Omega^t$ outside a $\Pc^u$-polar set we have for some universally measurable $\epsilon$-optimal selector $P_t^{\epsilon}$ that
\begin{align*}
&\E_{P_t^{\epsilon}(\omega)} \bigg[U_{t+1}\left((\omega, \cdot), V^{\pi_0,H,C }_{t}(\omega)+H_{t+1}(\omega)\Delta S_{t+1}(\omega, \cdot)-\mathds{1}_{\{t=T-1\}}\xi(\omega,\cdot)\right)\bigg]+\sum_{s=1}^{t} U(s,\omega,\Delta C_s(\omega))-\epsilon\\
&\le \ (-\epsilon)^{-1} \vee \bigg(\inf_{P \in \Pc^u_t(\omega)} \E_{P}\left[U_{t+1}\left((\omega, \cdot), V^{\pi_0,H,C}_{t}(\omega)+H_{t+1}(\omega)\Delta S_{t+1}(\omega, \cdot)-\mathds{1}_{\{t=T-1\}}\xi(\omega,\cdot)\right)\right]\\
&+\sum_{s=1}^{t} U(s,\omega,\Delta C_s(\omega))\bigg)\\
&\le \ (-\epsilon)^{-1} \vee \bigg(\sup_{(H',c') \in \mathcal{A}_{t,V_{t-1}^{\pi_0,H,C}(\omega)+H_t(\omega)\Delta S_t(\omega)}(\omega)}\inf_{P \in \Pc^u_t(\omega)} \E_{P}\bigg[U_{t+1}\big((\omega, \cdot), V^{\pi_0,H,C}_{t-1}(\omega)\\
&+H_t(\omega)\Delta S_t(\omega)-c'
+H'\Delta S_{t+1}(\omega,\cdot)-\mathds{1}_{\{t=T-1\}}\xi(\omega, \cdot)\big)\bigg]+\sum_{s=1}^{t-1} U(s,\omega,\Delta C_s(\omega))+U(t,c')\bigg)\\
&= \ (-\epsilon)^{-1} \vee \bigg( U_{t}(\omega, V^{\pi_0,H,C}_{t-1}(\omega)+H_t(\omega)\Delta S_t(\omega))+\sum_{s=1}^{t-1}U(s,\omega, \Delta C_s(\omega)) \bigg).
\end{align*}
Given $P \in \Pc^u$ we thus have
\begin{align*}
&\E_{P}\bigg[(-\epsilon)^{-1} \vee \bigg(U_{t}(V^{\pi_0,H,C}_{t-1}+H_t\Delta S_t)+\sum_{s=1}^{t-1}U(s,\Delta C_s)\bigg)\bigg] \\
\ge\ &\E_{P \otimes P_t^{\epsilon}}\bigg[U_{t+1}(V^{\pi_0,H,C}_{t}+H_{t+1}\Delta S_{t+1}-\mathds{1}_{\{t=T-1\}}\xi)+\sum_{s=1}^{t} U(s,\Delta C_s)\bigg]-\epsilon \\
\ge &\inf_{P' \in \Pc^u} \E_{P'}\bigg[U_{t+1}(V^{\pi_0,H,C}_{t}+H_{t+1}\Delta S_{t+1}-\mathds{1}_{\{t=T-1\}}\xi)+\sum_{s=1}^{t} U(s,\Delta C_s)\bigg]-\epsilon.
\end{align*}
As $\epsilon>0$ and $P \in \Pc^u$ were arbitrary \eqref{eq max} follows. Noting that $U_0(\pi_0)=\inf_{P \in \Pc^u}\E_{P}[U_0(V^{\pi_0,H,C}_0)]$ a repeated application of \eqref{eq max} yields
\begin{align*}
U_0(\pi_0) \ge \inf_{P \in \Pc^u} \E_{P}[ U_1(\pi_0+H_1 \Delta S_1)] \ge \dots &\ge \inf_{P \in \Pc^u} \E_{P}\bigg[U_T(V^{\pi_0,H,C}_{T-1}+H_T \Delta S_T-\xi)+\sum_{s=1}^{T-1} U(s, \Delta C_s)\bigg]\\
&=\inf_{P \in \Pc^u} \E_{P}\bigg[\sum_{s=1}^{T} U(s, \Delta C_s)\bigg].
\end{align*}
As $(H,C) \in \mathcal{A}_{\pi_0}$ was arbitrary, it follows that $U_0(\pi_0) \ge v(\pi_0)$. This concludes the proof, since $\pi_0=\pi(\xi)$.
\end{proof}

\subsection{Proof of \cref{thm. unique_new}}
\begin{proof}
Existence of an optimal investment consumption strategy follows from \cref{thm. exis_new}. We now show uniqueness of optimisers. We fix $0\le t \le T-1$ and recall the definition of $\tilde{U}_t$ given in \cref{lem usa}. Note that one can show that the function $$(\omega,P) \mapsto \sup_{(H,c) \in \mathcal{A}_{t,x}(\omega)} \E_{P} [\tilde{U}_{t+1}((\omega, \cdot), x+H \Delta S_{t+1}(\omega,\cdot)-c-\mathds{1}_{\{t=T-1\}}\xi(\omega, \cdot))]
+U(t, c)$$ is lower semianalytic by reducing the above expression to a supremum over a countable set as in the proof of \cref{lem usa}. Recall that again by \cref{lem usa} there exists a set of full $\Pc^u$ measure on which $\tilde{U}_t=U_t$ for all $0\le t \le T$. For the rest of the proof we take $\omega$ in the intersection of this set with $\Omega_{\text{NA}}^t$. Using the same Jankov-von-Neumann argument as in the proof of \cref{thm. exis_new} and \cref{cor. minim} we conclude that for each $t=0, \dots, T-1$ there exists a sequence $P^n_{t}:\Omega^{t}\to \mathfrak{P}(\Omega)$ of universally measurable kernels such that $P_{t}^n(\omega) \in \Pc^u_{t}(\omega)$ and for $x\ge \pi_t(\xi)(\omega)$
\begin{align*}
\sup_{(H,c) \in \mathcal{A}_{t,x}(\omega)} \E_{P^n_{t}(\omega)} [\tilde{U}_{t+1}((\omega, \cdot), x+H \Delta S_{t+1}(\omega,\cdot)-c-\mathds{1}_{\{t=T-1\}}\xi(\omega, \cdot))]
+U(t, c) \downarrow \tilde{U}_{t}(\omega,x). 
\end{align*}
Since $\Pc^u_t(\omega)$ is compact, there exists a probability measure $\hat{P}_t(\omega) \in \Pc^u_t(\omega)$ and a subsequence $\{n_k(\omega)\}_{k \in \N}$ such that $\lim_{k \to \infty} P^{n_k(\omega)}_t(\omega)=\hat{P}_t(\omega)$.
%a probability measure $\hat{P}=\hat{P}_0 \otimes\dots \otimes \hat{P}_{T-1}\in \Pc^u$ such that without loss of generality $P^n=P_0^n \otimes \dots \otimes P_{T-1}^n$ converges weakly to $\hat{P}$. Then by Assumption \cref{Ass 3} also $P_{t}^n(\omega)$ converges weakly to $\hat{P}_{t}(\omega)$ for $t=0, \dots, T-1$ for $\hat{P}$-a.e. $\o \in \Omega^t$. 
We now show, that for $\Pc^u$-q.e. $\omega \in \Omega^{t}$ and $x\ge\pi_{t}(\omega)$ the functions
\begin{align*}
U_{t}(\omega, x) &= \sup_{(H,c) \in \mathcal{A}_{t,x}(\omega)} \inf_{P \in \Pc^u_{t}(\omega)} \E_{P}[U_{t+1}((\omega, \cdot), x+H\Delta S_{t+1}(\omega, \cdot)-c-\mathds{1}_{\{t=T-1\}}\xi(\omega,\cdot))]\\&+ U(t, \omega,c)
\end{align*}
have a unique optimizer $(H,c)\in \mathcal{A}_{t,x}(\omega)$. For notational convenience we assume that $0\le t \le T-2$. We note that by concavity of $\tilde{U}_{t+1}$ and $U(t, \cdot)$ the function 
\begin{align*}
(H,c) \mapsto \inf_{P \in \Pc^u_{t}(\omega)}\ E_{P}\left(\tilde{U}_{t+1}\left((\omega, \cdot), y+H\Delta S_{t+1}(\omega,\cdot)-c\right)\right)+U(t,c)
\end{align*}
is concave. Now assume that there are $(H^1,c^1)$, $(H^2,c^2) \in \mathcal{A}_{t, x}(\omega)$ such that
\begin{align*}
&\inf_{P \in \Pc^u_{t}(\omega)} E_{P}\left(\tilde{U}_{t+1}\left((\omega, \cdot),x+H^1\Delta S_{t}(\omega,\cdot)-c^1\right)\right)+U(t,c^1)\\
=\ &\inf_{P \in \Pc^u_{t}(\omega)} E_{P}\left(\tilde{U}_{t+1} \left((\omega, \cdot), x+H^2\Delta S_{t}(\omega,\cdot)-c^2\right)\right)+U(t, c^2)\\=\ &\tilde{U}_{t}(\omega,x).
\end{align*}
Note that for the strategy $(H^3,c^3):=((H^1+H^2)/2,(c^1+c^2)/2) \in \mathcal{A}_{t, x}(\omega)$ we have by concavity
\begin{align*}
&\inf_{P \in \Pc^u_{t}(\omega)} E_{P}\left(\tilde{U}_{t+1}\left((\omega, \cdot),x+H^3\Delta S_{t}(\omega,\cdot)-c^3\right)\right)+U(t,c^3)\\
\ge\ &\frac{1}{2}\bigg(\inf_{P \in \Pc^u_{t}(\omega)} E_{P}\left(\tilde{U}_{t+1}((\omega, \cdot),\left(x+H^1\Delta S_{t}(\omega,\cdot)-c^1\right)\right)+U(t,c^1)\\
+\ &\inf_{P \in \Pc^u_{t}(\omega)} E_{P}\left(\tilde{U}_{t+1}\left((\omega, \cdot),y+H^2\Delta S_{t+1}(\omega,\cdot)-c^2\right)\right)+U(t,c^2)\bigg)=\tilde{U}_{t}(\omega,x).
\end{align*}
We thus conclude
\begin{align*}
\inf_{P \in \Pc^u_{t}(\omega)} E_{P}\left(\tilde{U}_{t+1}\left((\omega, \cdot),x+H^3\Delta S_{t}(\omega,\cdot)-c^3\right)\right)+U(t,c^3)=\tilde{U}_{t}(\omega,x).
\end{align*}
Furthermore, for any $x\ge \pi_t(\omega)$ and any maximizer $(\tilde{H}, \tilde{c}) \in \mathcal{A}_{t,x}(\omega)$ of $\tilde{U}_{t}(\omega,x)$ we have
\begin{align} \label{eq. limit1}
&\sup_{(H,c) \in \mathcal{A}_{t, x}(\omega)}\left( \E_{P^{n_k(\omega)}_{t}(\omega)} [\tilde{U}_{t+1}((\omega, \cdot), x+H \Delta S_{t+1}(\omega,\cdot)-c)]+U(t, c)\right)\nonumber \\
\ge\ &\E_{P^{n_k(\omega)}_{t}(\omega)}[\tilde{U}_{t+1}((\omega, \cdot),x+\tilde{H} \Delta S_{t+1}(\omega,\cdot)-\tilde{c})]+U(t, \tilde{c}) \\ 
\ge\ &\inf_{P \in \Pc^u_{t}(\omega)} \E_{P}[\tilde{U}_{t+1}((\omega, \cdot),x+\tilde{H}\Delta S_{t+1}(\omega,\cdot)-\tilde{c})]+U(t, \tilde{c})=\tilde{U}_{t}(\omega,x)\nonumber,
\end{align}
so taking limits in \eqref{eq. limit1} we find
\begin{align*}
\lim_{k \to \infty} \E_{P^{n_k(\omega)}_{t}(\omega)} [\tilde{U}_{t+1}((\omega, \cdot),x+\tilde{H} \Delta S_{t+1}(\omega,\cdot)-\tilde{c})]+U(t, \tilde{c})=\tilde{U}_{t}(\omega,x).
\end{align*}
Furthermore we note that by assumption and \cref{lem usa}  $\tilde{U}_{t}(\omega,y)$ is bounded by some $C$ on $\{(\omega,x) \in \Omega^t \times \R \ | \ x \ge \pi_t(\xi)(\omega)\}$, non-decreasing as well as continuous in $y$ and $\xi$ is continuous. Note the superhedging prices $\omega \mapsto \pi_t(\xi)(\omega)$ are continuous on $\{(\omega, v)\in  \Omega^t \ | \ v\in f_{t-1}(\omega)\}$ by assumption.\\
%that by \cref{Ass 3}.(2) we can apply \cref{prop:pathwiseex} to conclude that . 
For $n\in \N_+$  we define the shifted utility function $$U^{1/n}(T,x):=U(T,x+1/n).$$
Furthermore we inductively define the corresponding one-step versions for the multiperiod case $U^{1/n}_T(\omega,x):= U^{1/n}(T,x)$ and $$U^{1/n}_t(\omega,x):= \sup_{(H,c) \in \mathcal{A}_{t,x}(\omega)} \inf_{P \in \Pc^u_{t}(\omega)} \E_{P}[U^{1/n}_{t+1}((\omega, \cdot), x+1/n+H\Delta S_{t+1}(\omega, \cdot)-c)]+ U(t,c)$$ for $1\le t\le T-1$. Note that in particular $U^{1/n}(t,x)$ fulfils \cref{Ass 3} for all $n\in \N$ and $1\le t\le T$. Denote their lower semianalytic versions $\tilde{U}^{1/n}_t(\omega,x)$. Again by \cref{lem usa} there exists a set of full $\Pc^u$-measure, such that $\tilde{U}^{1/n}_t(\omega,x)=U^{1/n}_t(\omega,x)$ for all $n\in \N$ and $1\le t\le T$ and we fix $\omega$ in this set from now on. We now show by backwards induction that for all $n\in \N$ the function $(\omega,x) \mapsto \tilde{U}^{1/n}_{t}(\omega,x+1/n)$ is continuous in every point of the set $\{(\omega,x)\in \Omega^t\times \R \ | \ x\ge \pi_{t}(\xi)(\omega) \}$: Let us assume the hypothesis is true for $t+1$ and fix $n\in \N$, $x \ge \pi_{t}(\xi)(\omega)$. For any $(\tilde{\omega},\tilde{x}) \in \Omega^t\times \R$ we have
\begin{align*}
\left|\tilde{U}^{1/n}_{t}(\omega,x+1/n)-\tilde{U}^{1/n}_{t}(\tilde{\omega},\tilde{x}+1/n)\right|&\le \left|\tilde{U}_{t}^{1/n}(\omega,x+1/n)-\tilde{U}_{t}^{1/n}(\omega,\tilde{x}+1/n)\right|\\
&+\left|\tilde{U}_{t}^{1/n}(\omega,\tilde{x}+1/n)-\tilde{U}_{t}^{1/n}(\tilde{\omega},\tilde{x}+1/n)\right|.
\end{align*}
As $x \mapsto \tilde{U}^{1/n}_{t}(\omega,x+1/n)$ is continuous on $[\pi_t(\xi)(\omega)-1/n,\infty)$, there exists $\delta>0$ such that the first summand can be bounded by $\epsilon/2$ if $|x- \tilde{x}|\le \delta$. Thus it is sufficient to show that there exists $\tilde{\delta} >0$ such that for all $|\tilde{\omega}-\omega| \le \tilde{\delta}$ we have 
\begin{align*}
\left|\tilde{U}_{t}^{1/n}(\omega,\tilde{x}+1/n)-\tilde{U}_{t}^{1/n}(\tilde{\omega},\tilde{x}+1/n)\right| \le \epsilon/2.
\end{align*}
Indeed, note first that by \cref{rem:bounded} and the same contradiction argument as in the proof of \cref{prop:pathwiseex} choosing $\tilde{\delta}>0$ small enough we can assume that for any superhedging strategy $(H,c)\in \mathcal{A}_{t,\pi_{t}(\xi)(\tilde{\omega})}(\tilde{\omega})$ we have $|(H,c)| \le \tilde{C}$ for some $\tilde{C}>0$ independent of $\tilde{\omega}$. Furthermore we can choose $\tilde{\delta}>0$, such that $|\pi_t(\xi)(\omega)-\pi_t(\xi)(\tilde{\omega})|\le 1/n$.\\
%, thus in particular $\tilde{U}_t(\omega, \tilde{x}+\alpha)>-\infty$ for $\tilde{x}\ge \pi_t(\xi)(\tilde{\omega}).$\\
Next we make the following observation: As $\Pc^u_t(\omega)$ is weakly compact by assumption, there exists a compact set $[0,K]^d \subseteq \Omega$, such that $P(([0,K]^d)^c) \le \epsilon/(48C)$ for all $P \in \Pc^u_t(\omega)$. By the induction hypothesis $(v,y)\mapsto \tilde{U}^{1/n}_{t+1}(v,y+1/n)$ is continuous in every point of the set  $\{(v,y)\in \Omega^{t+1}\times \R \ | \ y\ge \pi_{t+1}(\xi)(v)\}$ and thus uniformly continuous on a compact subset. There exists $1/n>\delta_0>0$ such that for $v, \tilde{v}\in B_{1}(\omega)\times \{u\in \Omega \ | \ \inf_{\tilde{u}\in [0,K]^d}|u-\tilde{u}|\le \delta_0\}$, $y\in [\pi_{t+1}(\xi)(v), 2CK]$ and $|(v,y) -(\tilde{v},\tilde{y})|\le \delta_0$ we have
\begin{align}\label{eq. estimate2}
 \left|\tilde{U}^{1/n}_{t+1}(v,y+1/n)-\tilde{U}^{1/n}_{t+1}(\tilde{v},\tilde{y}+1/n)\right|\le \epsilon/24.
\end{align}
By \cref{Ass 3}.(1) and by adapting $\tilde{\delta}$ accordingly, for all $\tilde{\omega}\in \Omega^{t}$ such that $|\omega-\tilde{\omega}|<\tilde{\delta}$ and for all $P \in \Pc^u_{t}(\omega)$, there exists $\tilde{P} \in \Pc^u_{t}(\tilde{\omega})$ such that $d_L(P, \tilde{P}) \le \tilde{\epsilon}:=\delta_0/(2\tilde{C}) \wedge \epsilon/(48C)$.
It follows by Strassen's theorem that there exists a measure $\pi \in \mathfrak{P}(\R^d \times \R^d)$ and two random variables $X\sim P\circ ( S_{t+1})^{-1}(\tilde{\omega},\cdot)$ and $\tilde{X}\sim \tilde{P}\circ ( S_{t+1})^{-1}(\tilde{\omega},\cdot)$ such that $\pi(|X-\tilde{X}|\ge \tilde{\epsilon}) \le \tilde{\epsilon}$. Thus we conclude that for $y,\tilde{y}:\Omega\to \R$ with $|y(x) -\tilde{y}(\tilde{x})|\le \delta_0$ whenever $\pi_{t+1}(\tilde{\omega})\le \tilde{y}(\tilde{x})\le 2CK$  and $|x-\tilde{x}|\le \tilde{\epsilon}$ 
\begin{align}\label{eq. estimate}
&\left|\E_{P}\left[\tilde{U}^{1/n}_{t+1}((\tilde{\omega}, \cdot),1/n+y(\cdot))\right]- \E_{\tilde{P}}\left[\tilde{U}^{1/n}_{t+1}((\tilde{\omega}, \cdot),1/n+\tilde{y}(\cdot))\right]\right| \\
\ =\ &\left| \E_{\pi}\left[\tilde{U}^{1/n}_{t+1}((\tilde{\omega}, X),1/n+y(X))-\tilde{U}_{t+1}^{1/n}((\tilde{\omega}, \tilde{X}),1/n+\tilde{y}(\tilde{X}))\right]\right|\nonumber\\
\ \le&\ \ \E_{\pi}\bigg[\bigg|\tilde{U}^{1/n}_{t+1}((\tilde{\omega}, X),1/n+y(X))-\tilde{U}^{1/n}_{t+1}((\tilde{\omega}, \tilde{X}),1/n+\tilde{y}(\tilde{X}))\bigg| \mathds{1}_{\{X\in [0,K]^d, \ |X-\tilde{X}|\le \tilde{\epsilon}\}}\bigg]+\frac{C\epsilon}{12C}\nonumber\\
\ \le&\ \ \epsilon/12+\epsilon/12=\epsilon/6.\nonumber
\end{align}
Now we modify $\tilde{\delta}>0$ such that $|\pi_t(\xi)(\omega)-\pi_t(\xi)(\tilde{\omega})|\le \delta_0$ if $|\omega-\tilde{\omega}|\le \tilde{\delta}$.
Furthermore applying \cref{Thm Nutz1} for the function $(\omega,x+1/n)\mapsto \tilde{U}^{1/n}_{t}(\omega,x+1/n)$ there exists a maximiser $(H',c') \in \mathcal{A}_{t,\tilde{x}+1/n}(\tilde{\omega})$ of $$ \sup_{(H,c) \in \mathcal{A}_{t,\tilde{x}+1/n}(\tilde{\omega})} \inf_{P \in \Pc^u_{t}(\tilde{\omega})} \E_{P}[\tilde{U}^{1/n}_{t+1}((\tilde{\omega}, \cdot), \tilde{x}+1/n+H\Delta S_{t+1}(\tilde{\omega}, \cdot)-c)]+ U(t,c)$$ and a strategy $(H,c'-\beta) \in \mathcal{A}_{t,\tilde{x}+1/n}(\omega),$ where $\beta:=c'\wedge|\pi_t(\xi)(\omega)-\pi_t(\xi)(\tilde{\omega})|\le \delta_0/2$. Furthermore there exists $P \in \Pc_{t}^u(\omega)$ such that 
$$\tilde{U}^{1/n}_{t}(\omega,\tilde{x}+1/n) \ge \E_{P}\left[\tilde{U}^{1/n}_{t+1}((\omega, \cdot),\tilde{x}+2/n+H\Delta S_{t+1}(\omega, \cdot)-c'+\beta)\right]+U(t,c'-\beta)-\epsilon/6.$$
 Note that we can modify $\tilde{\delta}>0$ such that $|(\omega,HS_t(\omega))-(\tilde{\omega}, HS_t(\tilde{\omega}))|\le (\tilde{C}+2)\tilde{\delta}\le \delta_0/2$. Now by \eqref{eq. estimate2} with $y(\cdot)=\tilde{x}+1/n+H\Delta S_{t+1}(\omega,\cdot)-c'+\beta$ and $\tilde{y}(\cdot)=\tilde{x}+1/n+H\Delta S_{t+1}(\tilde{\omega},\cdot)-c'$
\begin{align*}
&\E_{P}[\tilde{U}^{1/n}_{t+1}((\omega, \cdot),\tilde{x}+2/n+H\Delta S_{t+1}(\omega, \cdot)-c'+\beta)]+U(t,c'-\beta)-\epsilon/6\\
\ge\  &\E_{P}[\tilde{U}^{1/n}_{t+1}((\tilde{\omega}, \cdot),\tilde{x}+2/n+H\Delta S_{t+1}(\tilde{\omega}, \cdot)-c')] +U(t,c')-\epsilon/3 
\end{align*}
follows and by \eqref{eq. estimate} with 
$y(\cdot)=\tilde{x}+1/n+H\Delta S_{t+1}(\tilde{\omega},\cdot)-c'$, 
$\tilde{y}(\cdot)=\tilde{x}+1/n+H'\Delta S_{t+1}(\tilde{\omega},\cdot)-c'$ and
noting that $|H-H'|\le 2\tilde{C}$
\begin{align*}
&\E_{P}\left[\tilde{U}^{1/n}_{t+1}((\tilde{\omega}, \cdot),\tilde{x}+2/n+H\Delta S_{t+1}(\tilde{\omega}, \cdot)-c')\right] +U(t,c')-\epsilon/3 \\
\ge\  &\E_{\tilde{P}}\left[\tilde{U}^{1/n}_{t+1}[(\tilde{\omega}, \cdot),\tilde{x}+2/n+H'\Delta S_{t+1}(\tilde{\omega}, \cdot)-c')\right]+U(t,c')- \epsilon/2 \\
\ge\ & \tilde{U}^{1/n}_{t}(\tilde{\omega},\tilde{x})- \epsilon/2.
\end{align*}
Exchanging the roles of $\omega$ and $\tilde{\omega}$ concludes the proof of the induction step. \\
This shows in particular continuity of $\omega' \mapsto \tilde{U}^{1/n}_{t+1}((\omega, \omega'), x+1/n+\tilde{H} \Delta S_{t+1}(\omega, \omega')-\tilde{c})$ as $\omega' \mapsto x+\tilde{H} \Delta S_{t+1}(\omega, \omega')-\tilde{c}$ is continuous.
As this function is also $\Pc^u_t(\omega)$-q.s. bounded by \cref{lem usa} (recall that $(\tilde{H},\tilde{c})\in \mathcal{A}_{t,x}(\omega)$), we conclude by use of the Portmanteau theorem that
\begin{align*}
\tilde{U}_{t}(\omega,x)&=\inf_{n\in \N}\tilde{U}_t^{1/n}(\omega,x)\\
&=\inf_{n\in \N}\liminf_{k \to \infty}  \E_{P^{n_k(\omega)}_{t}(\omega)} [\tilde{U}_{t+1}^{1/n}((\omega, \cdot),x+1/n+\tilde{H} \Delta S_{t+1}(\omega,\cdot)-\tilde{c})]+U(t, \tilde{c})\\
& \ge \inf_{n\in \N}\E_{\hat{P}_{t}(\omega)} [\tilde{U}^{1/n}_{t+1}((\omega, \cdot), x+1/n+\tilde{H} \Delta S_{t+1}(\omega,\cdot)-\tilde{c})]+U(t, \tilde{c})\\
&= \E_{\hat{P}_{t}(\omega)} [\tilde{U}_{t+1}((\omega, \cdot), x+\tilde{H} \Delta S_{t+1}(\omega,\cdot)-\tilde{c})]+U(t, \tilde{c})\\
&\ge \inf_{P \in \Pc^u_{t}(\omega)} \E_{P} [\tilde{U}_{t+1}((\omega, \cdot),x+\tilde{H} \Delta S_{t+1}(\omega,\cdot)-\tilde{c})]+U(t, \tilde{c}),
\end{align*}
which yields for $x\ge \pi_t(\omega)$
\begin{align*}
\tilde{U}_{t}(\omega,x) = \E_{\hat{P}_{t}(\omega)} [\tilde{U}_{t+1}((\omega,\cdot),x+\tilde{H} \Delta S_{t+1}(\omega,\cdot)-\tilde{c})]+U(t, \tilde{c}).
\end{align*}
%A similar argument shows
%\begin{align*}
%U_{t-1}(\omega,y)=\sup_{H \in \mathcal{H}_{t-1,Pi_{t-1}(\omega)}(\omega)} \E_{\hat{\P}_{t-1}(\omega)} (U_t((\omega, \cdot),y+H \Delta S_t(\omega,\cdot))).
%\end{align*}
In particular for $i=1,2$
\begin{align*}
&\E_{\hat{P}_{t}\omega)}\left[\tilde{U}_{t+1}\left((\omega, \cdot),x+H^3\Delta S_{t+1}(\omega,\cdot)-c^3\right)\right]+U(t,c^3)\\
= \ &\E_{\hat{P}_{t}(\omega)}\left[\tilde{U}_{t+1}\left((\omega, \cdot),x+H^i\Delta S_{t+1}(\omega,\cdot)-c^i\right)\right]+U(t,c^i).
\end{align*}
Now since
\begin{align*}
(H,c) \mapsto \E_{\hat{P}_{t}(\omega)} [\tilde{U}_{t+1}((\omega,\cdot),x+H \Delta S_{t+1}(\omega,\cdot)-c)]+U(t, c)
\end{align*}
is concave and strictly concave in $c$, we need to have $c^1=c^2$ and
\begin{align*}
H^1 \Delta S_{t+1}(\omega,\cdot)&=H^2\Delta S_{t+1}(\omega, \cdot) \quad \hat{P}_{t}(\omega)-\text{a.s}
\end{align*}
Lastly denote by $\Xi_t$ the correspondence
\begin{align*}
\Xi_t(\omega)=\left\{ P \in \mathcal{P}^u_t(\omega) \ \bigg| \ \tilde{U}_t(x,\omega)= \sup_{(H,c) \in \mathcal{A}_{t, x}(\omega)} \E_{P} [\tilde{U}_{t+1}((\omega, \cdot), x+H \Delta S_{t+1}(\omega,\cdot)-c)]+U(t, c) \right\}
\end{align*}
for $x\ge \pi_t(\omega)$ and note that by measurable selection arguments as in \cite{BN}[proof of Lemma 4.10, p. 848] the set
\begin{align*}
\left\{(\omega, P) \in \text{graph}(\Pc^u_t) \ \bigg| \ \sup_{(H,c) \in \mathcal{A}_{t, x}(\omega)} \E_{P} [\tilde{U}_{t+1}((\omega, \cdot), x+H \Delta S_{t+1}(\omega,\cdot)-c)]+U(t, c) - \tilde{U}_t(x,\omega) \le 0 \right\}
\end{align*}
is an element of $\textbf{A}(\Fut\otimes \mathcal{B}(\mathfrak{P}(\Omega)))$, where $\textbf{A}(\Fut\otimes \mathcal{B}(\mathfrak{P}(\Omega)))$ is the set of all nuclei of Suslin schemes on $\Fut\otimes \mathcal{B}(\mathfrak{P}(\Omega))$. In consequence there exists an $\Fc^{\mathcal{U}}_t$-measurable function $\hat{P}_t: \Omega^t \to \mathfrak{P}(\Omega)$ such that $\text{graph}(\hat{P}_t)\subseteq \text{graph}(\Xi_t)$. This concludes the proof.
\end{proof}
\begin{remark}
If we assume that $H^1-H^2 \in  \text{span}_{\hat{P}_{t}(\omega)}(\Delta S_{t+1}(\omega,\cdot))$, then $H^1=H^2$.
\end{remark}

\end{appendices}

\bibliographystyle{abbrvnat}
\bibliography{bibliomiklos4}

\begin{thebibliography}{46}
\providecommand{\natexlab}[1]{#1}
\providecommand{\url}[1]{\texttt{#1}}
\expandafter\ifx\csname urlstyle\endcsname\relax
  \providecommand{\doi}[1]{doi: #1}\else
  \providecommand{\doi}{doi: \begingroup \urlstyle{rm}\Url}\fi

\bibitem[Acciaio et~al.(2013)Acciaio, Beiglb{\"o}ck, Penkner, and
  Schachermayer]{ABPW13}
B.~Acciaio, M.~Beiglb{\"o}ck, F.~Penkner, and W.~Schachermayer.
\newblock A model-free version of the fundamental theorem of asset pricing and
  the super-replication theorem.
\newblock \emph{Math. Finance}, 26\penalty0 (2):\penalty0 233--251, 2013.

\bibitem[Aksamit et~al.(2018)Aksamit, Deng, Ob\l\'oj, and Tan]{Aksamitetal:18}
A.~Aksamit, S.~Deng, J.~Ob\l\'oj, and X.~Tan.
\newblock The robust pricing--hedging duality for american options in discrete
  time financial markets.
\newblock \emph{Math. Finance}, pages 1--37, 2018.
\newblock \doi{10.1111/mafi.12199}.
\newblock Published Online.

\bibitem[Aliprantis and Border(2006)]{Hitch}
C.~D. Aliprantis and K.~C. Border.
\newblock \emph{Infinite Dimensional Analysis~: A Hitchhiker's Guide}.
\newblock Grundlehren der Mathematischen Wissenschaften [Fundamental Principles
  of Mathematical Sciences]. Springer-Verlag, Berlin, 3rd edition, 2006.

\bibitem[Avellaneda et~al.(1996)Avellaneda, Levy, and Paras]{AP95}
M.~Avellaneda, A.~Levy, and A.~Paras.
\newblock Pricing and hedging derivatives securities in markets with uncertain
  volatilities.
\newblock \emph{Appl. Math. Finance}, 2\penalty0 (2):\penalty0 73--88, 1996.

\bibitem[Bartl(2019)]{Bart16}
D.~Bartl.
\newblock Exponential utility maximization under model uncertainty for
  unbounded endowments.
\newblock \emph{Ann. Apl. Prob.}, 29\penalty0 (1):\penalty0 577--612, 2019.

\bibitem[Bartl et~al.(2017)Bartl, Kupper, and Neufeld]{bartl2017pathwise}
D.~Bartl, M.~Kupper, and A.~Neufeld.
\newblock Pathwise superhedging on prediction sets.
\newblock \emph{arXiv preprint arXiv:1711.02764}, 2017.

\bibitem[Bayraktar and Zhou(2017)]{bayraktar2017arbitrage}
E.~Bayraktar and Z.~Zhou.
\newblock On arbitrage and duality under model uncertainty and portfolio
  constraints.
\newblock \emph{Math. Finance}, 27\penalty0 (4):\penalty0 988--1012, 2017.

\bibitem[Beiglb{\"o}ck and Nutz(2014)]{beiglbock2014martingale}
M.~Beiglb{\"o}ck and M.~Nutz.
\newblock Martingale inequalities and deterministic counterparts.
\newblock \emph{Electron. J. Probab.}, 19, 2014.

\bibitem[Beiglb\"{o}ck et~al.(2017)Beiglb\"{o}ck, Cox, Huesmann, Perkowski, and
  Pr\"omel]{BCHPP:17}
M.~Beiglb\"{o}ck, A.~M.~G. Cox, M.~Huesmann, N.~Perkowski, and D.~Pr\"omel.
\newblock Pathwise superreplication via vovk's outer measure.
\newblock \emph{Finance Stoch.}, 21\penalty0 (4):\penalty0 1141--1166, 2017.

\bibitem[Bertsekas and Shreve(2004)]{bs}
D.~P. Bertsekas and S.~Shreve.
\newblock \emph{Stochastic Optimal Control: The Discrete-Time Case}.
\newblock Athena Scientific, 2004.

\bibitem[Biagini et~al.(2017)Biagini, Bouchard, Kardaras, and
  Nutz]{biagini2014robust}
S.~Biagini, B.~Bouchard, C.~Kardaras, and M.~Nutz.
\newblock Robust fundamental theorem for continuous processes.
\newblock \emph{Math. Finance}, 27\penalty0 (4):\penalty0 963--987, 2017.

\bibitem[Black and Scholes(1973)]{BS73}
F.~Black and M.~Scholes.
\newblock The pricing of options and corporate liabilities.
\newblock \emph{J. Polit. Econ.}, 81\penalty0 (3):\penalty0 637--654, 1973.

\bibitem[Blanchard and Carassus(2017)]{BC16}
R.~Blanchard and L.~Carassus.
\newblock Multiple-priors investment in discrete time for unbounded utility
  function.
\newblock \emph{Ann. Apl. Prob.}, 2017.

\bibitem[Bonnice and Reay(1969)]{bonnice1969relative}
W.~E. Bonnice and J.~R. Reay.
\newblock Relative interiors of convex hulls.
\newblock \emph{Proc. Amer. Math. Soc.}, 20\penalty0 (1):\penalty0 246--250,
  1969.

\bibitem[Bouchard and Nutz(2015)]{BN}
B.~Bouchard and M.~Nutz.
\newblock Arbitrage and duality in nondominated discrete-time models.
\newblock \emph{Ann. Apl. Prob.}, 25\penalty0 (2):\penalty0 823--859, 2015.

\bibitem[Bouchard et~al.()Bouchard, Deng, and Tan]{bouchard2017super}
B.~Bouchard, S.~Deng, and X.~Tan.
\newblock Super-replication with proportional transaction cost under model
  uncertainty.
\newblock \emph{Math. Finance (to appear)}.

\bibitem[Burzoni et~al.()Burzoni, Frittelli, Hou, Maggis, and Obloj]{Bur16b}
M.~Burzoni, M.~Frittelli, Z.~Hou, M.~Maggis, and J.~Obloj.
\newblock Pointwise arbitrage pricing theory in discrete time.
\newblock \emph{Math. Oper. Res. (to appear)}.

\bibitem[Burzoni et~al.(2016{\natexlab{a}})Burzoni, Frittelli, and
  Maggis]{Burz16a}
M.~Burzoni, M.~Frittelli, and M.~Maggis.
\newblock Model-free superhedging duality.
\newblock \emph{Ann. Apl. Prob.}, 2016{\natexlab{a}}.

\bibitem[Burzoni et~al.(2016{\natexlab{b}})Burzoni, Frittelli, and
  Maggis]{ClassS}
M.~Burzoni, M.~Frittelli, and M.~Maggis.
\newblock Universal arbitrage aggregator in discrete-time markets under
  uncertainty.
\newblock \emph{Finance Stoch.}, 20\penalty0 (1-50), 2016{\natexlab{b}}.

\bibitem[Carassus and Vargiolu(2018)]{CV17}
L.~Carassus and T.~Vargiolu.
\newblock Super-replication price: it can be ok.
\newblock \emph{ESAIM Proc. Surv.}, 65:\penalty0 241--281, 2018.

\bibitem[Carassus et~al.(2006)Carassus, Gobet, and Temam]{CGT06}
L.~Carassus, E.~Gobet, and E.~Temam.
\newblock A class of financial products and models a class of financial
  products and models where super-replication prices are explicit.
\newblock In \emph{International Symposium on Stochastic Processes and
  Mathematical Finance}. Ritsumeikan University, Tokyo, 2006.

\bibitem[Cox and Ob\l\'{o}j(2011)]{CO11}
A.~Cox and J.~Ob\l\'{o}j.
\newblock Robust pricing and hedging of double no-touch options.
\newblock \emph{Finance Stoch.}, 15\penalty0 (3):\penalty0 573--605, 2011.

\bibitem[Cox et~al.(1979)Cox, Ross, and Rubistein]{CRR79}
J.~Cox, S.~Ross, and M.~Rubistein.
\newblock Option pricing: a simplified approach.
\newblock \emph{J. Financial Econ.}, 7\penalty0 (229-264), 1979.

\bibitem[Davis and Hobson(2007)]{DaHo07}
M.~Davis and D.~Hobson.
\newblock The range of traded option prices.
\newblock \emph{Math. Finance}, 17\penalty0 (1):\penalty0 1--14, 2007.

\bibitem[Delbaen and Schachermayer(2006)]{DelSch05}
F.~Delbaen and W.~Schachermayer.
\newblock \emph{The Mathematics of Arbitrage}.
\newblock Springer Finance, 2006.

\bibitem[Denis and Kervarec(2013)]{denis2007utility}
L.~Denis and M.~Kervarec.
\newblock Utility functions and optimal investment in non-dominated models.
\newblock \emph{SIAM J. Control Optim.}, 51\penalty0 (3):\penalty0 1803--1822,
  2013.

\bibitem[Denis and Martini(2006)]{DenisMartini:06}
L.~Denis and C.~Martini.
\newblock A theoretical framework for the pricing of contingent claims in the
  presence of model uncertainty.
\newblock \emph{Ann. Appl. Probab.}, 16\penalty0 (2):\penalty0 827--852, 2006.
\newblock ISSN 1050-5164.
\newblock \doi{10.1214/105051606000000169}.

\bibitem[Dupire(2010)]{Dupire}
B.~Dupire.
\newblock Talk and discussion during the "robust techniques in quantitative
  finance" conference, oxford-man institute of quantitative finance.
\newblock March 2010.

\bibitem[Epstein and Ji(2014)]{EJ13}
L.~G. Epstein and S.~Ji.
\newblock Ambigous volatility, possibility and utility in continuous time.
\newblock \emph{J. Math. Econ.}, 50:\penalty0 269--282, 2014.

\bibitem[F{\"o}llmer and Kramkov(1997)]{FK97}
H.~F{\"o}llmer and D.~Kramkov.
\newblock Optional decompositions under constraints.
\newblock \emph{Probab. Theory Related Fields}, 109\penalty0 (1025), 1997.

\bibitem[F{\"o}llmer and Schied(2002)]{fs}
H.~F{\"o}llmer and A.~Schied.
\newblock \emph{Stochastic Finance: An Introduction in Discrete Time}.
\newblock Walter de Gruyter \& Co., Berlin, 2002.

\bibitem[Hobson(1998.)]{Ho982}
D.~Hobson.
\newblock Robust hedging of the lookback option.
\newblock \emph{Finance Stoch.}, 2:\penalty0 329--347, 1998.

\bibitem[Hou and Ob\l\'oj(2018)]{HO18}
Z.~Hou and J.~Ob\l\'oj.
\newblock Robust pricing-hedging dualities in continuous time.
\newblock \emph{Finance Stoch.}, 22\penalty0 (3):\penalty0 511--567, 2018.

\bibitem[Knight(1921)]{Kni}
F.~Knight.
\newblock \emph{Risk, Uncertainty, and Profit}.
\newblock Boston, MA: Hart, Schaffner Marx; Houghton Mifflin Co, 1921.

\bibitem[Lange(1973)]{Lange}
K.~Lange.
\newblock Borels sets of probability measures.
\newblock \emph{Pacific J. Math.}, 48\penalty0 (1):\penalty0 141--161, 1973.

\bibitem[Lyons(1995)]{Ly95}
F.~Lyons.
\newblock Uncertain volatility and the risk-free synthesis of derivatives.
\newblock \emph{Appl. Math. Finance}, 2:\penalty0 117--133, 1995.

\bibitem[Merton(1973)]{merton_no_dominance1973}
R.~C. Merton.
\newblock Theory of rational option pricing.
\newblock \emph{Bell J. Econ.}, 4\penalty0 (1):\penalty0 141--183, 1973.

\bibitem[Neufeld and Sikic(2018)]{NS16}
A.~Neufeld and M.~Sikic.
\newblock Robust utility maximization in discrete time with friction.
\newblock \emph{SIAM J. Control Optim.}, 56\penalty0 (3):\penalty0 1912--1937,
  2018.

\bibitem[Nutz(2016)]{Nutz}
M.~Nutz.
\newblock Utility maximisation under model uncertainty in discrete time.
\newblock \emph{Math. Finance}, 26\penalty0 (2):\penalty0 252--268, 2016.

\bibitem[Obl\'{o}j and Wiesel(2018)]{jw}
J.~Obl\'{o}j and J.~Wiesel.
\newblock A unified framework to modelling financial markets in discrete time.
\newblock \emph{arXiv preprint arXiv:1808.06430}, 2018.

\bibitem[R\'asonyi and Stettner(2006)]{RS06}
M.~R\'asonyi and L.~Stettner.
\newblock On the existence of optimal portfolios for the utility maximization
  problem in discrete time financial models.
\newblock \emph{In: Kabanov, Y.; Lipster, R.; Stoyanov,J. (Eds), From
  Stochastic Calculus to Mathematical Finance, Springer.}, pages 589--608,
  2006.

\bibitem[Rockafellar and Wets(1998)]{rw}
R.~T. Rockafellar and R.~J.-B. Wets.
\newblock \emph{Variational analysis}.
\newblock Grundlehren der Mathematischen Wissenschaften [Fundamental Principles
  of Mathematical Sciences]. Springer-Verlag, Berlin, 1998.
\newblock ISBN 3-540-62772-3.

\bibitem[Sainte-Beuve(1974)]{bv}
M.-F. Sainte-Beuve.
\newblock On the extension of von {N}eumann-{A}umann's theorem.
\newblock \emph{J.\ Functional Analysis}, 17\penalty0 (1):\penalty0 112--129,
  1974.

\bibitem[Samuelson(1969)]{Samuelson}
P.~A. Samuelson.
\newblock Lifetime portfolio selection by dynamic stochastic programming.
\newblock \emph{Rev. Econom. Statist.}, 51:\penalty0 239--246, 1969.

\bibitem[Schied and Wu(2005)]{schied2005duality}
A.~Schied and C.-T. Wu.
\newblock Duality theory for optimal investments under model uncertainty.
\newblock \emph{Statist. Decisions}, 23\penalty0 (3/2005):\penalty0 199--217,
  2005.

\bibitem[Terkelsen(1973)]{terk}
F.~Terkelsen.
\newblock Some minimax theorems.
\newblock \emph{Math. Scand.}, 31\penalty0 (2):\penalty0 405--413, 1973.

\end{thebibliography}
\end{document}